\RequirePackage{ifthen}

\newcommand{\typof}{1} %
\newcommand{\longv}[1]{\ifthenelse{\equal{\typof}{0}}{}{#1}}
\newcommand{\shortv}[1]{\ifthenelse{\equal{\typof}{0}}{#1}{}}

\newcommand{\longshortv}[2]{\ifthenelse{\equal{\typof}{0}}{#2}{#1}}
\newcommand{\drop}[0]{\ifthenelse{\equal{\typof}{0}}{}{}}

\documentclass[a4paper, 10pt]{article}

\usepackage{etex}
\usepackage{rotating}
\usepackage{graphicx}
\usepackage{tabularx}
\usepackage{amsfonts}
\usepackage{amsthm}% or letter or a5paper or ... etc
\usepackage{amsmath}% or letter or a5paper or ... etc
\usepackage{mathabx2}
\usepackage{proof}
\usepackage{mathdots}
\usepackage{subcaption}
\usepackage{hyperref}
\usepackage{xcolor}
\usepackage{adjustbox}
\usepackage{amssymb}
\usepackage{color}

\usepackage{enumerate}

\usepackage{tikz}
\usepackage{mathrsfs}
\usepackage{mathabx}
\usepackage{lscape}
\usepackage{stmaryrd}
\usepackage{bussproofs}

\usepackage{tikz}
\usetikzlibrary{positioning,shapes,shapes.multipart}
\usetikzlibrary{cd,arrows,snakes,backgrounds}

\newtheorem{definition}{Definition}[section]

\newtheorem{notation}{Notation}
\newtheorem{theorem}{Theorem}[section]

\newtheorem{corollary}{Corollary}[section]

\newtheorem{lemma}[theorem]{Lemma}
\newtheorem{proposition}[theorem]{Proposition}
\newtheorem{remark}{Remark}[section]
\usepackage{xspace}

\newcommand{\B}[1]{\mathbf{#1}}

\newcommand{\SF}[1]{\mathsf{#1}}

\newcommand{\TT}[1]{\mathtt{#1}}

\newcommand{\C}[1]{\mathcal{#1}}
\newcommand{\BB}[1]{\mathbb{#1}}

\newcommand{\OV}[1]{\mathbin{\overline{#1}}}
\newcommand{\F}[1]{\mathfrak{#1}}

\newcommand{\To}{\Rightarrow}
\newcommand{\model}[1]{\llbracket#1 \rrbracket}

\newcommand{\monuno}{\mathsf{DyIL_{\To,\forall^{1}}}}

\newcommand{\MON}{\mathrm{1Mon}^{\To,\forall}}

\newcommand{\Bool}{\mathsf{Bool}}
\newcommand{\Nat}{\mathsf{Nat}}

\newcommand{\bb}[1]{{\color{black}#1}}
\newcommand{\rr}[1]{{\color{black}#1}}

\newcommand{\Fone}{{\mathrm{F}_{1}}}
\newcommand{\STLC}{{\mathrm{ST\lambda C}}}

\newcommand{\Nda}{{\mathrm{F_{at}}}}

\newcommand{\Ndac}{{\mathrm{F_{at}^{\clubsuit}}}}

  \newcommand{\rank}[1]{\mathsf{r}(#1)}
  \newcommand{\LET}[3]{\mathsf{let}\ #1 \ \mathsf{be} \ #2 \ \mathsf{in} \ #3}

  \newcommand{\IOand}{\mathsf{IO}^{\times}}
  \newcommand{\IOor}{\mathsf{IO}^{+}}

  \newcommand{\ML}{\mathrm{ML}}

\definecolor{color0}{HTML}{4682B4}

\usepackage[numbers]{natbib}
\usepackage{url}

\usepackage{varwidth}
\newsavebox{\mypti}
\newsavebox{\mypto}
\newsavebox{\mypta}
\newsavebox{\myptu}
\newsavebox{\mypte}

\EnableBpAbbreviations

\title{What's Decidable about (Atomic) Polymorphism?\\ (Extended Version)}
\author{Paolo Pistone\footnote{Universit\`a di Bologna, \texttt{paolo.pistone2@unibo.it}} \and Luca Tranchini\footnote{Eberhard Karls Universit\"at T\"ubingen, \texttt{luca.tranchini@gmail.com}}
}

\date{}

\begin{document}

\maketitle

\begin{abstract}
Due to the undecidability of most type-related properties of System F like type inhabitation or type checking, restricted polymorphic systems have been widely investigated (the most well-known being ML-polymorphism). 
In this paper we investigate System Fat, or atomic System F, a very weak predicative fragment of System F whose typable terms coincide with the simply typable ones. 
We show that the type-checking problem for Fat is decidable and we propose an algorithm which sheds some new light on the source of undecidability in full System F. 
Moreover, we investigate free theorems and contextual equivalence in this fragment, and we show that the latter, unlike in the simply typed lambda-calculus, is undecidable.
\end{abstract}

%
%
%\title{The undecidability of atomic polymorphism}%\keywords{This is not required.}
%%\subjclass{This is not required.}
%\author{Paolo Pistone, Luca Tranchini, Mattia Petrolo}
%\date{}
%\revauthor{Paolo, Pistone}
%\address{Wilhelm-Schickard-Institut, Universit\"at T\"ubingen\\
%Sand 13,
%D-72076 \\ Tübingen (Germany) 
%}
%\email{paolo.pistone@uni-tuebingen.de}
%
%\author{Luca Tranchini}
%\revauthor{Luca, Tranchini}
%\address{Wilhelm-Schickard-Institut, Universit\"at T\"ubingen\\
%Sand 13,
%D-72076 \\ Tübingen (Germany) 
%}
%\email{luca.tranchini@gmail.com}
%
%\author{Mattia Petrolo}
%\revauthor{Mattia, Petrolo}
%\address{Wilhelm-Schickard-Institut, Universit\"at T\"ubingen\\
%Sand 13,
%D-72076 \\ Tübingen (Germany) 
%}
%\email{luca.tranchini@gmail.com}

%\paragraph{Keywords} Atomic polymorphism, instantiation overflow, Russell-Prawitz translation, naturality condition, Yoneda lemma, invertible connective, second order logic

%\tableofcontents

\section{Introduction}
% !TEX root = Undecidability.tex

Polymorphism has been a central topic in programming language theory since the late sixties. 
Today, most general purpose programming languages employ some kind of polymorphism.
At the same time, under the Curry-Howard correspondence, quantification over types corresponds to quantification over propositions, that is, to second-order logic. In particular, 
System F, the 
archetypical  
type system for polymorphism, can be seen as a proof-system for (the $\To,\forall$-fragment of) second-order intuitionistic logic.

%- Polymorphism is very attractive: 
%it is at the heart of several concrete programming languages

In spite of the numerous applications of polymorphism, practically all interesting type-related properties of (Curry-style) System F (e.g.~type checking, type inhabitation, etc.) are undecidable, making this language impractical for any reasonable implementation.
This is one of the reasons why a wide literature has investigated more manageable subsystems of System F. 
Notably, \emph{$\ML$-}\emph{polymorphism} \cite{Milner1978, Milner1985, Milner} has found much success due to its  decidable type-checking.

Another direction of research was that of investigating \emph{predicative} subsystems of System~F \cite{Leivant1989, Leivant1991, Leivant1994, Danner99}. In particular, the so-called \emph{finitely stratified polymorphism} \cite{Leivant1991} yields a  stratification of System F through a sequence of predicative systems $(\mathrm F_{n})_{n\in \BB N}$ of growing expressive power (notably, $F_{0}$ is the simply typed $\lambda$-calculus $\STLC$, and 
$\ML$-polymorphism coincides with the rank-1 part of $\Fone$). 
Yet, in spite of such limitations, type checking becomes undecidable already at level 1 of this hierarchy \cite{Schubert2012}.

Could one tell exactly at which point, in the range from the simply typed $\lambda$-calculus and $\ML$ to full System F, the type-related properties of polymorphism become undecidable?

\subparagraph*{Atomic Polymorphism}
In more recent times Ferreira et al.~have undertaken the investigation of what can be seen as the least expressive predicative fragment of F, System $\Nda$, or \emph{atomic} System F \cite{Ferreira2013, Ferreira2009, Ferreira2015, Ferreira2016, Ferreira2018, Ferreira2020, EspiritoSanto2020b}.  The predicative restriction of $\Nda$ is such that a universally quantified type $\forall X.A$ can be instantiated solely with an atomic type, i.e.~a type variable.
In this way $\Nda$ sits in between level 0 (i.e.~$\STLC$) and level 1 of the finitely stratified hierarchy. 
Actually, $\Nda$ can be seen as a \emph{type refinement} system (in the sense of \cite{Mellies2015}) of $\STLC$, since all terms typable in $\Nda$ are simply typable (cf.~Lemma \ref{lemma:stlc}). 

In spite of its very limited expressive power, Ferreira et al.~have shown that, thanks to polymorphism, $\Nda$ enjoys some proof-theoretic properties that $\STLC$ lacks. In particular, they defined a predicative variant of the usual encoding of sum and product types inside F, yielding an embedding of intuitionistic propositional logic inside $\Nda$. 
However, while propositional logic is decidable, 
provability in second-order propositional intuitionistic logic, even with the atomic restriction, is undecidable \cite{Sobolev}. This argument (as recently observed in \cite{Protin2021}) can be extended to show that the \emph{type inhabitation} property, which is decidable for $\STLC$, is undecidable for $\Nda$.

%
%
%
%As the decidability of propositional logic implies the decidability of \emph{type inhabitation} in $\STLC$, it was recently observed \cite{Protin2021} that 
%

%an  
%undecidability argument for System F from \cite{CH} also applies to $\Nda$,
%% 
%% provability (i.e.~type inhabitation, see below) is undecidable also in $\Nda$ 
% (i.e.~to the $\forall,\to$-fragment of $\mathrm{LJ2_{at}}$).

%atomic second-order intuitionistic   type inhabitation in System F from \cite{CH} also applies to $\Nda$  \cite{Protin2021}.

\paragraph*{Contributions}

In this paper we investigate the following type-related properties of System $\Nda$:
%
%In spite of the aforementioned results, the decidability of most type-related properties of $\Nda$ is still open. In this paper we provide answers to some of these problems, yielding new insights on the broader question of understanding where the source of undecidability in full System F lies.
%We will consider the following properties:
\begin{center}
\begin{tabular}{ l l }
Type inhabitation (TI):  & given $A$, is there $t$ such that $\vdash t:A$? \\
Type-checking (TC): & given $\Gamma, A, t$, does $\Gamma\vdash t:A$?\\
Typability (T): & given $\Gamma, t$, is there $A$ such that $\Gamma\vdash t:A$?\\
Contextual equivalence (CE): & given $A, t,u$ such that $\vdash t,u:A$, 
do $\TT C[t]$ and $\TT C[u]$ reduce \\
& to the same boolean,
for all context $\TT C[\ ]: A\To \mathsf{Bool}$?
% \\
%   &  do $vt$ and $vu$ reduce to the same boolean?
\end{tabular}

\end{center}
 \noindent
 In Fig.~\ref{fig:table} we sum up what is already known and what is established in this paper (in bold) about such properties in predicative fragments of System F.
  Our main results are that in $\Nda$ (TC) and (T) are both decidable, and that (CE) is decidable if one restricts oneself to numerical functions, and undecidable in the general case.

Several decidability properties of $\Nda$ are tight, meaning that they all fail already for $\Fone$. 
%In particular, to the best of our knowledge, $\Nda$ is the \emph{only} known first-class fragment of F admitting a decidable type-checking (if one excludes $\ML$-style systems in which first-class polymorphism is encoded through type-coercions, e.g.~\cite{Remy2003, Sulzmann2007}).
In these cases, our arguments can be used to shed some new insights on the broader question of understanding where the source of undecidability for such properties in full System F lies.

\begin{figure}
\begin{center}
\adjustbox{scale=0.9, center}{
\begin{tabular}{ |c || c c c c c|}
\hline 
 &  $\mathrm F_{0}=\STLC$ &  $\Nda$ & $\ML$ & $\Fone$ & F \\
 \hline\hline
TI &
decidable \cite{Statman1979} & undecidable \cite{Protin2021} & decidable & undecidable \cite{CH} & undecidable \cite{Lob1976} \\
\hline
TC &decidable  \cite{Hindley1969} & \textbf{decidable}& decidable \cite{Milner} & undecidable \cite{Schubert2012}&  undecidable \cite{Wells}
\\ \hline
T  &decidable  \cite{Hindley1969} & \textbf{decidable} & decidable \cite{Milner} &  undecidable \cite{Schubert2012}&  undecidable \cite{Wells}
\\ \hline
$
\begin{matrix}
\text{CE} \\ 
\text{\small (for numerical} \\ \text{\small functions)}
\end{matrix}
$
 &decidable \cite{Padovani}
& \textbf{decidable} & \textbf{undecidable} &  \textbf{undecidable} &  undecidable$^{*}$  \\
 \hline
$
\begin{matrix}
\text{CE} \\ 
\text{\small (full)}
\end{matrix}
$
 &decidable  \cite{Padovani} &  \textbf{undecidable}&  \textbf{undecidable}&  \textbf{undecidable}&  undecidable$^{**}$\\
\hline
\end{tabular}
}
\end{center}
\caption{Decidable and undecidable properties of System F and some predicative fragments (in bold the properties established in the present paper).\\
{\small
$^{*}$: easy consequence of Rice's theorem and the typability of all primitive recursive functions in F.\\
$^{**}$: consequence of the undecidability of (CE) for numerical functions.
}}
\label{fig:table}
\end{figure}

\paragraph*{Plan of the paper}
After recalling the syntax of F and its fragment in Curry-style and Church-style, we address the properties (TI), (TC), (T) and (CE).

\subparagraph*{Type Inhabitation}
In Section \ref{sec:undeci} we shortly discuss the undecidability of (TI), by showing how the argument in \cite{CH} for System F applies to $\Nda$ too. This argument yields an encoding inside $\Nda$ of an undecidable fragment of \emph{first-order} intuitionistic logic. We also observe that $\Nda$ is actually equivalent to a first-order system, namely to the $\To,\forall$-fragment $\MON$ of first-order \emph{monadic} intuitionistic logic in a language with a unique monadic predicate. 
To our knowledge, the undecidability of $\MON$ has not been previously observed (although some slightly more expressive fragments - e.g.~including a primitive disjunction \cite{Gabbay1981} or  finitely many monadic predicates \cite{Schubert2016} - have been proven undecidable).

% 

%
%- A more general approach is given by Leivant's stratified hierarchy ...  in particular ML is the rank-1 fragment of 
%the first level $F_1$ of Leivant's hierarchy
%
%- However, already for $F_{1}$ type checking and typability are undecidable. A natural question is where the source of undecidability lies. 
%
%- In recent times Ferreira et al. have started the investigation of what can be seen as the least expressive predicative fragment of F, System Fat. Actually Fat is strictly less expressive than ANY fragment in Leivant's hierarchy.
%However, they have shown that Fat is expressive enough that one can embed IPC in it.
%
%- The decidability of most type-related properties of Fat is still open. In this paper we provide answers to these questions.
%Notably, we show that the study of Fat highlights where things start to ``go wrong'' with polymorphism. 
%Moreover, while most undecidability results exploit encodings of first-order systems, Fat is ITSELF be seen as a first order system.
%In fact we show that
%
%(1) Fat and $F_{1}$ are equivalent at the level of provability (and in particular both undecidable by an argument of Urzyczyn and Soerensen). Hence predicative instantiations are not necessary to make provability undecidable. 
%Undecidability is a standard property of FIRST-ORDER systems. 

\subparagraph*{Type-Checking and Typability}
%We then consider the type checking problem. 

In Section \ref{sec:deci} we consider the type-checking problem. 
The undecidability of (TC) for System F was established by Wells in \cite{Wells}, and was later extended to all predicative systems $\mathrm F_{n}$, for $n>0$ \cite{Schubert2012}. In all these cases this result was obtained by reducing an undecidable variant of \emph{second-order unification} (SOU) to the type-checking problem. On the other hand, the decidability of (TC) for $\ML$ (and $\mathrm F_{0}=\STLC$) is based on the famous Hindley-Milner algorithm \cite{Milner}, which reduces this problem to \emph{first-order unification} (FOU), which is decidable.   

%This is where the picture starts to change: we show that (TC) is decidable for (Curry-style) $\Nda$. Our decision algorithm is somehow illuminating on where 
%(TC) and (T) start to become undecidable.

The fundamental source of undecidability of SOU is 
the presence of \emph{cyclic} dependences between second order variables, expressed in the simplest case by equations of the form $X(t) = f(v_{1},\dots, v_{k-1},X(u), v_{k+1},\dots,v_{n})$. In fact, \emph{acyclic} SOU is decidable \cite{Jordi1998}.
When type-checking polymorphic programs, such cyclic dependencies are generated by \emph{self-applications}, i.e.~terms of the form $\lambda \vec x. x t_{1}\dots t_{k-1}x t_{k+1}\dots t_{n}$. In fact, in this case the type $\forall X.A$ assigned to the variable $x$ must satisfy a cyclic equation of the form 
$$A[X\mapsto C_{1}]= B_{1}\To \dots \To B_{k-1}\To A[X\mapsto C_{2}]\To B_{k+1}\To \dots \To B_{n}$$
(where $C_{1},C_{2}$ are suitable type instantiations of $X$). 
By constrast, no term containing a self-application can be typed in $\STLC$, since cyclic equations cannot be solved by FOU.

Since the terms typable in $\Nda$ can also be typed in $\STLC$, it follows that   
self-applications cannot be typed in $\Nda$ either. Using this observation, 
we describe a type-checking algorithm for $\Nda$ which works in two phases: first, it checks (using FOU) the presence of cyclic dependencies, and returns $\mathtt{failure}$ if it detects one; then, if phase 1 succeeds, it applies (a suitable variant of) acyclic SOU to decide type-checking. 
From the decidability of (TC), we deduce the decidability of (T) by a standard argument (see \cite{Barendregt92}).

\subparagraph*{Contextual Equivalence}

Studying the typable terms of $\Nda$ might not seem very interesting from a computational viewpoint, 
as these terms are already typable in $\STLC$. 
However, due to the presence of some form of polymorphism, investigating programs in $\Nda$ can be interesting  for \emph{equational} reasoning, as we investigate in Sections \ref{sec:equivalence0} and \ref{sec:equivalence}. In standard type systems, beyond the standard notions of program equivalence arising from the operational semantics (i.e.~$\beta\eta$-equivalence), there may exist several other {congruences} arising from either denotational models or from some notion of \emph{contextual equivalence}.
In $\STLC$, it is well-known that $\simeq_{\beta\eta}$ coincides with the congruence induced by any infinite extensional model \cite{Statman1982}, as well as with several notions of contextual equivalence (see \cite{Barendregt2013}, \cite{Dosen2001}).
%
%
%
%
% \cite{}: two programs $t,u:A$ are equated when for all context $\TT C[\ ]:A \vdash O$ (where $O$ is a fixed type of \emph{observables}, $\TT C[t]\simeq_{\beta\eta}\TT C[u]:O$.  
%We will here focus on two standard variants of contextual equivalence, the one given by $O=\Bool=\forall X.X\To X\To X$, noted $\simeq_{\Bool}$, and the one given by $O=\Nat=\forall X.(X\To X)\To (X\To X)$, noted $\simeq_{\Nat}$.
In polymorphic type systems the picture is rather different, since $\beta\eta$-equivalence is usually weaker than the congruences arising from extensional models (see \cite{Bainbridge1990,Hasegawa2009}), and also weaker than standard notions of contextual equivalence.
Moreover, while $\beta\eta$-equivalence is decidable, contextual equivalence is undecidable.
%
%
%It is well-known that in System F usual $\beta\eta$-equivalence does not coincide with contextual equivalence, that is, the coarsest congruence on typable terms which does not equate booleans and/or natural numbers. In fact, the latter equates much more terms than the former. Moreover, unlike $\beta\eta$-equivalence, which is decidable as a consequence of the strong normalization property of System F, contextual equivalence is undecidable. 
%
Since in many practical situations (see \cite{Voigt2008, Ahmed2017}) it is more convenient to reason up to notions of equivalence stronger than $\beta\eta$-equivalence,
several techniques to compute (approximations of) contextual equivalence have been investigated, e.g.~\emph{free theorems} \cite{Wadler1989}, \emph{parametricity} \cite{Reynolds1983}, \emph{dinaturality} \cite{Bainbridge1990}. 

Our investigation of contextual equivalence starts in Section \ref{sec:equivalence0} with an exploration of equational reasoning in $\Nda$ using free theorems.
%
%, when two programs of $\Nda$ are equivalent as programs in F, they are contextually equivalent also in $\Nda$. This justifies the use of free theorems to deduce equivalences in $\Nda$. 
% By reasoning in this way
% 
We show that the predicative encodings of sum and product types of Ferreira et al.~produce products and coproducts in $\Nda$ in the categorical sense, provided terms are considered up to (CE) (a fact which is known to hold in F for the usual, impredicative, encodings \cite{Hasegawa2009, StudiaLogica}). 
We then investigate (CE) for typable numerical functions. Using the fact that the primitive recursive functions are \emph{uniquely} defined in System F up to (CE), 
%(that is, for all terms $t,u$ defining the same function, one can prove that $t$ is equivalent to $u$), 
we show that (CE) for the representable numerical functions is decidable in $\Nda$, and undecidable in $\ML$.
Such results rely on the observation that (CE) becomes undecidable as soon as some super-polynomial function (like bounded multiplication) becomes representable.  
 As a consequence, (CE) is undecidable in all fragments $F_{n}$, for $n>0$, of the finitely stratified hierarchy as well.
% 
%To the best of our knowledge, the decidability of (CE) for predicative fragments of F has not been previously investigated. %We show that (CE) is undecidable in $\Nda$, $\ML$ and, as a consequence, in all fragments $\mathrm F_{n}$ with $n>0$. For $\Nda$, our argument is based on a reduction to (CE) of the type inhabitation problem of $\Nda$.  

Finally, in Section \ref{sec:equivalence}
we establish that (CE) is undecidable also in $\Nda$, by showing that the type inhabitation problem for a suitable extension of $\Nda$ can be reduced to it.
This result, together with the previous ones, shows that there is no hope to get a decidable contextual equivalence for polymorphic programs through a predicative restriction, and one has rather to look for other kinds of restrictions (see for instance \cite{Pistone2021}).

\section{Predicative Polymorphism and System $\Nda$}\label{sec:pre}
%\input{axioms}

% !TEX root = Undecidability.tex

%\subsection{Two Systems for Predicative Polymorphism}
The systems we consider in this paper are all restrictions of usual Church-style and Curry-style System F. The types are defined in both cases by the grammar 
$$
A,B::= X\mid A\To B \mid \forall X.A
$$
starting from a countable set $\TT{Var}^{2}$ of type variables $X_{1},X_{2},\dots$. 
The terms of Church-style System F are defined by the grammar below:
\begin{align*}
t^{A},u^{A}& ::= 
x^{A}\mid (\lambda x^{A}.t^{B})^{A\To B} \mid t^{B\To A}u^{B}\mid (\Lambda X.t^{A})^{\forall X.A} \mid (t^{\forall X.A}C)^{A[C/X]}
\end{align*}
For readability, we will often omit type annotations, when these can be guessed from the context.
The terms of Curry-style System F are standard $\lambda$-terms, with typing rules defined as in Fig.~\ref{fig:fcurry}, where $\Gamma$ indicates a partial function from term variables to types with a finite support, and by $X\notin \mathrm{FV}(\Gamma)$ we indicate that $X$ does not occur free in any type in $\mathrm{Im}(\Gamma)$.
We will call the type $C$ occurring in $(t^{\forall X.A}C)^{A[C/X]}$ and in the rule $\forall$E in Fig.~\ref{fig:fcurry} the \emph{witness} of the type instantiation.

We will indicate \emph{term contexts} (i.e.~terms with a unique occurrence of the \emph{hole} $[\ ]$) as 
$\TT C[\ ], \TT D[\ ]$. Moreover, we let $\TT C[\ ]: A \vdash B$ be a shorthand for $x\mapsto A\vdash \TT C[x]:B$.

%
%
%In the variant \emph{\`a la Church} the terms of System F are obtained by adding to standard $\lambda$-term constructors the \emph{type abstraction} $\Lambda X.t$ and \emph{type application} $tA$ constructors, $\beta$-reduction is enriched with the rule $(\Lambda X.t)A\longrightarrow_{\beta}t[A/X]$, and $\eta$-expansion is enriched with the rule 
%$t \longrightarrow_{\eta} (\Lambda X.t)X$. The typing rules are recalled in Fig.~\ref{fig:fchurch}.
%In the variant \emph{\`a la Curry} the terms of System F are standard $\lambda$-terms, and the typing rules 
%$\forall$I$_{\mathrm{Ch}}$ and $\forall$E$_{\mathrm{Ch}}$ are replaced by the rules
%$\forall$I$_{\mathrm{Ch}}$ and $\forall$E$_{\mathrm{Ch}}$ in Fig.~\ref{fig:fcurry}. 

%
%\begin{figure}
%jfnasjfs
%\caption{Type rules for System F \emph{\`a la Church}.}
%\label{fig:fchurch}
%\end{figure}
\begin{figure}
\fbox{
\begin{minipage}{\linewidth}
$$
\begin{matrix}
\AXC{$\Gamma(x)= A$}
\RL{$\mathrm{Var}$}
\UIC{$\Gamma \vdash x: A$}
\DP
\\ 
\\
%
%\end{align*}
%\begin{align*}
\AXC{$\Gamma, x\mapsto A\vdash t: B$}
\RL{$\mathrm{Abs}$}
\UIC{$\Gamma \vdash \lambda x.t: A\To B$}
\DP
\qquad 
\AXC{$\Gamma \vdash t: A\To B$}
\AXC{$\Gamma \vdash u:A$}
\RL{$\mathrm{Appl}$}
\BIC{$\Gamma \vdash tu: B$}
\DP\\
  \\
\AXC{$\Gamma\vdash t: B$}
\AXC{$X\notin \mathrm{FV}(\Gamma)$}
\RL{$\forall$I}
\BIC{$\Gamma \vdash t:\forall X.A$}
\DP
 \qquad 
\AXC{$\Gamma \vdash t: \forall X.A$}
\RL{$\forall$E}
\UIC{$\Gamma \vdash t: A[C/X]$}
\DP
\end{matrix}
$$
\end{minipage}
}
\caption{Typing rules for Curry-style System F.}
\label{fig:fcurry}
\end{figure}

System F is impredicative: \emph{any} type can figure as a witness. In particular, one can construct ``circular'' instantiations, in which a term of type $\forall X.A$ is instantiated with the same type as witness.
A \emph{predicative} fragment of System F is one in which witnesses are restricted in such a way to avoid such circular instantiations. %Any such fragment comes in both \emph{\`a la Church} and \emph{\`a la Curry} flavor.

We will focus on three predicative fragments of System F, both in Church- and Curry-style. 
The first is System $\Fone$, which is the fragment of F in which witnesses are \emph{quantifier-free}. 
%More precisely, Church-style $\Fone$ is obtained by restricting $(t^{\forall X.A}C)A[C/X]$ to $C$ being quantifier-free type, and 
%Curry-style $\Fone$ is obtained by restricting the $\forall$E rule to $C$ being a quantifier-free type.
The second is System $\Nda$, which is the fragment of F in which witnesses are \emph{atomic}, that is, type variables. 
%which is defined in a similar way but with type instantiations being type variables.
%
%System $\Fone$ is actually the first level of a more general {stratification} of System F in stronger and stronger predicative fragments \cite{Leivant1991}. 
%System $\Nda$ is a yet more restricted system, that has been recently explored since \cite{Ferreira2013}.
The third is system $\ML$ \cite{Milner1978, Milner}, which essentially coincides with the \emph{rank 1} fragment of $\Fone$.
For any type $A$, the {rank} $\rank A$ is the maximum number of nesting between $\To$ and $\forall$, and is defined inductively by $\rank X=0$, $\rank{A\To B}= \max\{\rank{A}+1, \rank{B}\}$ and $\rank{\forall X_{1}\dots X_{n}.A}= \rank{A}+1$ (where $n>0$ and $A$ does not start with a quantifier).
%
%\begin{cases}  \rank{A}+1 & \text{ if }A\neq \forall Y.A' \\ \rank{A} & \text{ otherwise}    \end{cases}$.
To define $\ML$ (since type-checking is decidable in $\ML$, we limit ourselves to Curry-style) one first has to enrich the set of $\lambda$-terms with the $\mathsf{let}$-constructor, and add a rule
$$
\AXC{$\Gamma, x\mapsto A\vdash t:B$}
\AXC{$\Gamma \vdash u:A$}
\RL{$\mathrm{let}$}
\BIC{$\Gamma \vdash \LET{x}{u}{t}:B$}
\DP
$$
$\ML$ is the fragment of the resulting system in which typing rules only contain judgements $\Gamma \vdash t:A$, where $\rank A\leq 1$ and for all $B\in \mathrm{Im}(\Gamma)$, $\rank B\leq 1$.

Observe that in $\Fone$ one can encode $\LET{x}{u}{t}$ by $(\lambda x.t)u$, so that the rule above becomes derivable. This is not possible in $\ML$, due to the rank restriction.
%\begin{remark}
%Due to the rank restriction, for all types $A,B$ of $\ML$, the type $A\To B$ need not be a type of $\ML$. For this reason the usual syntax of $\ML$ includes a $\mathsf{let}$ constructor, so that if $\Gamma, x:A\vdash t:B$ and $\Gamma \vdash u:A$ hold, then $\Gamma \vdash \LET{x}{u}{t}: B$ holds.
%\end{remark}
%
%
%
%- Introduce (Curry) system F, stratification and Fat
%
%	- Recall System F Curry and Church
%	- Predicativity: restrict forall E-rule to avoid circular instances
%	- F1 and Fat--> quantifier-free instances/atomic instances. 
%	- A more general stratification...

\subparagraph*{Impredicative and Predicative Encodings}

It is well-known that sum and product types can be encoded inside System F by 
\begin{align*}
A \widetilde{+}B&= \forall X.(A\To X)\To (B\To X)\To X \\
A \widetilde{\times}B& = \forall X.(A\To B\To X)\To X
\end{align*}
where the type variable $X$ is fresh. The encoding of term constructors $\iota_{i}(\cdot)$, $\langle \cdot, \cdot\rangle$ 
and term destructors $\mathrm{Case}_{C}(\cdot, x^{A}.\cdot, x^{B}.\cdot)$ and $\pi_{i}(\cdot)$ 
is given (in Church-style) by:
\begin{align*}
\iota_{1}(t)  &= \Lambda X.\lambda f^{A\To X}. \lambda g^{B\To X}.ft 
& \qquad \mathrm{Case}_{C}(t, x^{A}.u, x^{B}.v) & =  tC( \lambda x^{A}.u)(\lambda x^{B}.v)\\
  \iota_{2}(t)  & = \Lambda X.\lambda f^{A\To X}. \lambda g^{B\To X}.gt 
 & \qquad \pi_{1}(t)& = t A \lambda x^{A}.\lambda y^{B}.x\\ 
\langle t,u \rangle &  = \Lambda X. \lambda f^{A\To B\To X}.ftu
& \qquad\pi_{2}(t)  & = t B \lambda x^{A}.\lambda y^{B}.y
\end{align*}
%The encoding of term destructors $\mathrm{Case}_{C}(\cdot, x^{A}.\cdot, x^{B}.\cdot)$ and $\pi_{i}(\cdot)$ is given by:
%\begin{align*}
%\end{align*}
At the level of provability, the encoding is \emph{faithful}: a type is inhabited in the extension of System F with sum and product types iff the encoded type is inhabited in System F. Moreover, the encoding of $\widetilde +$ satisfies the  \emph{disjunction property}: $A\widetilde +B$ is inhabited iff either $A$ or $B$ are inhabited. 

At the level of conversions, the encoding translates $\beta$-reduction rules for sum and product types into $\beta$-reduction in F. Instead, the $\eta$-rules for sums and products are not translated by the $\beta$- and $\eta$- rules of System F. Yet, the equivalence generated by $\beta$- and $\eta$-rules is preserved by \emph{contextual equivalence} in System F (more on this in Section \ref{sec:equivalence0}).

The encoding of sum and product types is impredicative: the encoding of term destructors requires witnesses of arbitrary complexity. Notably, given a term $t$ of type $A\widetilde + B$, the term 
$\mathrm{Case}_{A\widetilde +B}(t, x^{A}. \iota_{1}(x), x^{B}.\iota_{2}(x))$, of type $A\widetilde + B$, has a circular instantiation of $A\widetilde + B$. 

In \cite{Ferreira2013} an alternative, predicative, encoding was defined within System $\Nda$. The fundamental observation is that 
the unrestricted $\forall$E rule is derivable from the restricted one for the types of the form $A\widetilde + B$ and $A\widetilde \times B$ (the authors call this phenomenon \emph{instantiation overflow}). In fact\footnote{We essentially follow here an encoding from \cite{EspiritoSanto2020b}.}, for any type $C$ of System F one can define contexts $ \IOor_{C}[ \ ]:A\widetilde + B\vdash (A\To C)\To (B\To C)\To C$ and 
$\IOand_{C}[\ ]: A\widetilde \times B\vdash (A\To B\To C)\To C$ by induction on $C$:
\begin{align*}
\IOor_{X}[\ ]&=\IOand_{X}[\ ] = [\ ]X \\
\IOor_{C_{1}\To C_{2}}[\ ]& = \lambda f^{A\To C_{1}\To C_{2}}.\lambda g^{B\To C_{1}\To C_{2}}.\lambda y^{C_{1}}. \IOor_{C_{2}}[\ ] (\lambda z^{A}.fzy)(\lambda z^{B}.gzy)\\
\IOand_{C_{1}\To C_{2}}[\ ]& = \lambda f^{A\To B\To C_{1}\To C_{2}}.\lambda y^{C_{1}}. \IOor_{C_{2}}[\ ] (\lambda z^{A}\lambda w^{B}.fzwy)\\
\IOor_{\forall Y.C'}[\ ]& = \lambda f^{A\To \forall Y.C'}.\lambda g^{A\To \forall Y.C'}.\Lambda Y. \IOor_{C'}[\ ] (\lambda z^{A}.fzY)(\lambda z^{B}.gzY) \\
\IOand_{\forall Y.C'}[\ ]& = \lambda f^{A\To B\To \forall Y.C'}.\Lambda Y. \IOor_{C'}[\ ] (\lambda z^{A}.\lambda w^{B}.fzwY)
\end{align*}
One can thus encode the type destructors as for F, by replacing the type application $xC$ in $\mathrm{Case}_{C}(t, x^{A}.u,x^{B}.v)$ with either $\IOor_{C}[x]$ or 
$\IOand_{C}[x ]$.

At the level of provability, this embedding is {faithful} when restricted to {simple types}, i.e.~for the intuitionistic propositional calculus (see \cite{Ferreira2015}): a simple type (possibly containing finite sums and products) is inhabited iff its encoding is inhabited in $\Nda$. However, faithfulness does not hold for the extension of $\Nda$ with sum and product types (see \cite{SL2}). In particular, one can construct types $C,D$ of F such that $C\widetilde + D$ is inhabited in $\Nda$ while $C+D$ is not inhabited in the extension of $\Nda$ with sums and products.
This also implies that the disjunction property fails for  $C\widetilde + D$ in $\Nda$, since neither $C$ nor $D$ are inhabited.

Interestingly, at the level of conversions, this encoding is stronger than the usual one: it translates not only $\beta$-reductions, but also the 
permutative conversions and a restricted form of $\eta$-conversion for sums, into
 $\beta$ and $\eta$-reductions of $\Nda$ (see \cite{Ferreira2009, Ferreira2016b, EspiritoSanto2020b}).
%
%
%- Instantiation overflow: quick overview. Embedding IPC--> Fat. Observe that it is NOT faithful.
%
%	- Usual translation of IPC in F is IMPREDICATIVE: you need circular instances to translate or-E
%	- InstaOv: the translated formulas impredicatve instances are derivable
%	- Failure of faithfulness and of disjunction property. 
%	
%	

\section{Type Inhabitation}\label{sec:undeci}
%\input{undenew}
% !TEX root = Undecidability.tex

In this section we discuss type inhabitation in the systems $\Nda$ and $\Fone$. We briefly recall the undecidability argument for (TI) in System F from \cite{CH}, and observe that this applies to $\Nda$ (a more detailed reconstruction can be found in \cite{Protin2021}).

The argument in \cite{CH} (which was later simplified in \cite{Dudenhefner}) is based on an embedding inside F of an undecidable fragment of first-order logic. 
We recall the argument in a few more details, so that it will be clear that the same argument shows the undecidability of type inhabitation in both $\Nda$ and $\Fone$.

We consider the $\To,\forall$-fragment of \emph{dyadic} intuitionistic first-order logic $\monuno$, that is, in a language with no function symbol and a finite number of at most binary relation symbols. We consider sequents of the form $\Gamma\vdash \bot$ where $\Gamma$ consists of three type of assumptions:
\begin{enumerate}[i.]
\item atomic formulas different from $\bot$; \label{itemi}
\item closed formulas of the form $\forall \vec \alpha. ( \varphi_{1}\To \dots \To \varphi_{n}\To \psi)$, where $\varphi_{1},\dots, \varphi_{n},\psi$ are atomic formulas and each variable in $\psi$ occurs in some the $\varphi_{i}$;\label{itemii}
\item closed formulas of the form $\forall \alpha(\forall \beta(p(\alpha,\beta)\To \bot)\To \bot)$.
\label{itemiii}
\end{enumerate}

We fix a finite number of distinguished type variables:
\begin{itemize}
\item for each relation symbol $p$, three variables $p_{1},p_{2},p_{3}$;
\item five more variables $\spadesuit,\bullet, \circ_{1},\circ_{2},\star$.
\end{itemize}
We let, for any type $A$, $A^{\bullet}:=A\To \bullet$, and we define, for all types $A,B$: 
\begin{align*}
p_{AB}&= (A^{\bullet}\To p_{1})\To (B^{\bullet}\To p_{2})\To p_{3} \\
p(A,B)&= p_{AB}\To \star
\end{align*}
For any type $A$, we let $\mathcal U(A)$ be the set of all types $(A^{\bullet}\To p_{i})\To \circ_{1}, A^{\bullet}\To \circ_{2}$, where $i=1,2$. Given a finite list of types  $A_{1},\dots, A_{n}$, we let $\mathcal U(A_{1},\dots, A_{n})\To B$ be a shorthand for $C_{1}\To \dots \To C_{k}\To B$, where $C_{1},\dots, C_{k}$ are the types in 
$\bigcup_{i} \mathcal U(A_{i})$. 

Each formula $\varphi$ of $\monuno$ is translated into a type $\OV \varphi$ as follows:
\begin{align*}
\OV{p(\alpha_{i},\alpha_{j})} & = p(X_{i},X_{j})  \qquad  \OV\bot= \spadesuit \\
\OV{\varphi \To \psi}&= \OV\varphi \To \OV \psi \\
\OV{\forall \alpha_{i}.\varphi} & = \forall \vec X_{i}.(\mathcal U(X_{1},\dots, X_{n}) \To \OV \varphi)
\end{align*}

One can easily check the following by induction:
\begin{proposition}\label{prop:urzy0}
If $\varphi_{1},\dots, \varphi_{n} \vdash \varphi$ is provable in  $\monuno$ and $\alpha_{i_{1}},\dots, \alpha_{i_{k}}$ are the variables that occur in $\mathrm{FV}(\varphi)$ but not in $\mathrm{FV}(\varphi_{1},\dots,\varphi_{n})$, then 
$x_{1}\mapsto\OV\varphi_{1},\dots, x_{n}\mapsto\OV\varphi_{n}, \vec y\mapsto\mathcal U(X_{i_{1}},\dots, X_{i_{k}}) \vdash t: \OV\varphi$ holds in $\Nda$ for some term $t$.
\end{proposition}

The less trivial part is the following:

\begin{theorem}[\cite{CH}, Theorem 11.6.14]\label{th:urzy}
If  $\varphi_{1},\dots, \varphi_{n}$ satisfy \ref{itemi}-\ref{itemiii}, and  
$x_{1}\mapsto\OV\varphi_{1},\dots, x_{n}\mapsto\OV\varphi_{n},\vec y\mapsto \mathcal U(X_{i_{1}},\dots, X_{i_{k}}) \vdash t: \spadesuit$ is deducible in System F, then
$\varphi_{1},\dots, \varphi_{n}\vdash \bot $ is provable in $\monuno$.
\end{theorem}

Since $\Nda$ and $\Fone$ are both fragments of F, we can freely substitute them for System F in the statement of Theorem \ref{th:urzy}. Then, together with Proposition \ref{prop:urzy0} we deduce:

\begin{corollary}\label{cor:urzy}
(TI) is undecidable in both $\Nda$ and $\Fone$.
\end{corollary}

%More precisely, Corollary \ref{cor:urzy} follows from the equivalence of provability in $\Nda$ and in its extension with a finite number of type constants, which is established in App.~\ref{appB}.

\begin{remark}
Although $\Nda$ and $\Fone$ are both undecidable, they are not equivalent at the level of provability. For instance, the type
$(\forall X.X\To Y)\To (X\To X)\To Y$ is inhabited in $\Fone$ (by the term $\lambda x^{\forall X.X\To Y}.\lambda y^{X\To X}.x(X\To X)y$), but not in $\Nda$ (as easily seen by a proof-search argument).
\end{remark}

\begin{remark}
The undecidability of the atomic fragment of (full) second-order intuitionistic logic has been known since (at least) \cite{Sobolev}. However, from this one cannot deduce the undecidability of $\Nda$, due to the fact that disjunction is not faithfully definable in $\Nda$ (see also \cite{SL2}). 
\end{remark}

\begin{remark}
It is not difficult to see that System $\Nda$ is equivalent to a first-order system, namely to the $\To,\forall$-fragment $\mathsf{1Mon}_{\To,\forall}$ of \emph{monadic} first-order intuitionistic logic in the language with no function symbol and a \emph{unique} monadic predicate. The equivalence is given by an obvious bijection between formulas and types given by 
$\widehat{p(\alpha_{i})}=X_{i}$, $\widehat{\varphi\To \psi}=\widehat\varphi\To \widehat \psi$ and $\widehat{\forall \alpha_{i}.\varphi}= \forall X_{i}.\widehat\varphi$. 
Hence, a consequence of Corollary \ref{cor:urzy} is that provability in $\mathsf{1Mon}_{\To,\forall}$ is undecidable. 
Provability in extensions of $\mathsf{1Mon}_{\To,\forall}$ with either finitely many monadic predicates, or with disjunction, is known to be undecidable \cite{Gabbay1981, Schubert2012}. To the best of our knowledge, the undecidability of $\mathsf{1Mon}_{\To,\forall}$ had not been observed before.

\end{remark}

%
%- Remark: undecidability with or is well-known since Sobolev, but it does not imply undecidability of Fat since the embedding of IPC is not faithful,
%
%- Remark: equivalence of Fat and 1Mon, deduce undecidability of 1Mon (this is new, while Mon is undecidable and 
%1Mon with or is undecidable)

\section{Typability and Type-checking}\label{sec:deci}
%\input{decide}
% !TEX root = Undecidability.tex

In usual implementations of polymorphic type systems the Church-style type discipline is generally considered inconvenient, due to the heavy amount of type annotations. Instead, Curry-style languages, for which a compiler can (either completely or partially) reconstruct type annotations, are generally preferred (two standard examples are the languages $\mathsf{ML}$ and $\mathsf{Haskell}$). 
This is the reason why type-checking algorithms for polymorphic type systems in Curry-style (or in some variants of Curry-style with \emph{partial} type annotations \cite{Pfenning1993}) have been extensively investigated \cite{Henglein89, Urzy93, Wells, Schubert2012}.

However, while $\ML$ admits a decidable type checking in Curry-style (a main reason for its success), type checking  has been shown to be undecidable
 for System F and most of its variants (including the predicative system $\Fone$ \cite{Schubert2012}), making the Curry-style version of such systems impractical for implementation.

For the simply typed $\lambda$-calculus (and crucially also for $\ML$), the type-checking problem can be reduced to \emph{first-order unification} (FOU), that is, to the problem of unifying first-order terms (in a language with a unique binary function symbol corresponding to $\To$). Typically, an application $tu:b$ will produce a first-order equation of the form 
$a_{t}= a_{u}\To b$, where $a_{t},a_{u}$ are variables indicating the type of $t$ and the type of $u$, respectively. 
As FOU is decidable, this suffices to show that type-checking is decidable in this case.

In the case of polymorphism FOU is not sufficient to solve type-checking. In fact, already in $\Fone$ one can type terms, like e.g.~$\lambda x.xx$, which contain \emph{self-applications}. Using FOU, this term yields the unsolvable equation 
$a_{x}=a_{x}\To b$, showing that $\lambda x.xx$ is not typable in either $\STLC$ or $\ML$. To type-check System F programs one can replace FOU with either \emph{semi-unification} \cite{Henglein89, Urzy93} or \emph{second order unification} (SOU) \cite{Pfenning1993, Schubert2012}.
Here we focus on the latter: in SOU one tries to unify equations involving terms constructed from first-order variables $a,b,c,\dots$ as well as second order variables $\SF F,\SF G,\dots$. For instance, the term $\lambda x.xx$ above yields the equations
\begin{align}\label{eq:cyclic}
\SF F a & = (\SF Fb)\To \SF G
\end{align}
where $\forall X.A=\forall X.\SF FX$ indicates the type of $x$, and the variables $a,b$
encode the possible witnesses which permit to type $xx$ (in Church-style one could indicate this with
$\lambda x^{\forall X.\SF FX} . ((xa)^{\SF Fa=\SF Fb\To \SF G}(xb)^{\SF Fb})^{\SF G}$). For instance, a (non-unique) solution to this problem is given by 
$\SF F\mapsto \lambda x.x$,  $\SF G\mapsto Z$, $ a\mapsto Y\To Z$, $b \mapsto Y$.%, corresponding to the explicit typing $(\lambda x:\forall X.X) (x(Y\to Z))(xY)$. 

 Unfortunately, SOU is undecidable \cite{Goldfarb1981}. Moreover, one can encode restricted (but still undecidable) variants of SOU in the type checking problem for $\Fone$ \cite{Schubert2012}, showing that (TC) is undecidable for $\Fone$.
A fundamental ingredient of these undecidability arguments is the appeal to \emph{variable cycles} (see the discussion in \cite{Jordi1998}) like the one in Eq.~\eqref{eq:cyclic}, that is, to unification problems from which one can deduce an equation of the form $\SF Fa_{1}\dots a_{n}= u[ \SF F]$ equating a second-order variable $\SF F$ with some term containing $\SF F$ itself.

Conversely, \emph{acyclic} SOU, that is, the problem of unifying SOU problems containing no variable cycles, is decidable \cite{Jordi1998}. These observations can be used to show that type-checking is actually decidable in $\Nda$. 
In fact, a fundamental property of $\Nda$ (and a reason for its very limited expressive power) is that any term typable in $\Nda$ is already typable in the simply-typed $\lambda$-calculus. Indeed, the following is easily checked by induction:

\begin{lemma}\label{lemma:stlc}
If $\Gamma \vdash t:A$ is derivable in the Curry-style $\Nda$, then $|\Gamma |\vdash t:|A|$ is derivable in the simply typed $\lambda$-calculus, where $|A|$ is defined by
$
|X|=o$, $|A\To B|=|A|\To |B|$, $|\forall X.A|=|A|$, and $|\Gamma|(x)=|\Gamma(x)|$.
\end{lemma}

%From Lemma \ref{lemma:stlc} we can easily infer that the typability problem is decidable:
%
%\begin{corollary}
%(T) for Curry-style $\Nda$ is decidable (in polynomial time).
%\end{corollary}
%\begin{proof}
%(T) is decidable in polynomial time for $\STLC$ \cite{}. Now, (T) for $\Nda$ can be reduced to $\STLC$ by observing that, given $\Gamma,t$, on the one hand, if $\Gamma\vdash t:A$ holds in $\Nda$ for some type $A$, then $|\Gamma|\vdash t:|A|$ holds in $\STLC$ and, on the other hand, if 
%$|\Gamma| 
%
%\end{proof}

An immediate consequence of Lemma \ref{lemma:stlc} is that one cannot type $\lambda x.xx$ in $\Nda$ and, more generally,  that any $\lambda$-term that would give rise to a variable cycle cannot be typed in $\Nda$.
Observe that the converse does not hold: from the fact that $|\Gamma |\vdash t:|A|$ holds, one cannot deduce $\Gamma \vdash t:A$ (take for instance $t=x$, $\Gamma(x)=X$ and $A=\forall X.X$).  

However, these observations suggest that type checking for $\Nda$ can be decided by reasoning in two phases: 
to check if $\Gamma \vdash t:A$ is derivable in $\Nda$, first check if $|\Gamma|\vdash t:|A|$ is derivable in $\STLC$ using FOU;
%
%
%
%given the type-checking problem $\mathrm{TC}({\Nda})(\Gamma, t,A)$ (i.e.~``is $\Gamma \vdash t:A$ derivable in $\Nda$?''), first solve $\mathrm{TC}(\STLC)(|\Gamma|, t, |A|)$ using FOU;
 if this first step fails, then the original problem must fail; if the first step succeeds, then the original type-checking problem for $\Nda$ yields an instance of (a suitable variant of) acyclic SOU, which must be decidable. By reasoning in this way, one can thus establish:
\begin{theorem}
(TC) for Curry-style $\Nda$ is decidable.
\end{theorem}

In App.~\ref{appB} we describe in detail the decision algorithm for $\Nda$, which is based on a variant of second-order unification, that we call \emph{$\Nda$-unification}.
The fundamental idea is to consider SOU problems in a language with first-order \emph{sequence variables} $a,b,\dots$ and two kinds of second-order variables: \emph{projection variables} $\alpha,\beta,\dots$ and \emph{type schemes} $\SF F, \SF G,\dots$. The intuition is that a term of the form $\alpha a_{1}\dots a_{n}$ describes a (skolemized) witness; since the witnesses in $\Nda$ are type variables, solving for $\alpha$ means associating it with either a constant function or a projection. Instead, a term of the form $\SF F\F a_{1}\dots \F a_{n}$ stands for the application of suitable witnesses $\F a_{1},\dots, \F a_{n}$ to some type $\SF F$, hence solving for $\SF F$ means associating it with some function $\lambda X_{1}\dots X_{n}.A(X_{1},\dots, X_{n})$, where $A(X_{1},\dots, X_{n})$ is some type expression parametric on the type variable $X_{1},\dots, X_{n}$.  
Hence, for example, checking if $\Gamma \vdash xy: \forall Z.Z$ holds in $\Nda$, where $\Gamma(x)=\forall X.X\To X$ and $\Gamma(y)= \forall Y.Y$, yields the equations
%
%$$
%\SF F \F a_{1}\dots \F a_{n}=u 
%$$
%where $\F a_{i}$ are either sequence variables $a$ (i.e.~variables standing for a sequence of type variables - whose length must be guessed through the unification process), or an expression of the form $\alpha a_{1}\dots a_{n}$, which encode a possible type application ( formally, the variable $\alpha$ stands for some function that computes from $a_{1},\dots, a_{n}$ a type variable). So for example checking if the term $\lambda xy.xy$
%type-checks with $\Gamma(x)= \forall X.X\To X$, $\Gamma(y)=\forall Y.Y$ and $A=\forall Z.Z$ yields the equation
\begin{align*}
\SF F X & = X \To X & \qquad
\SF G Y & = Y \\
\SF F (\alpha Z) & = \SF G (\beta Z)  \To \SF H Z& \qquad
\SF H Z & = Z
\end{align*}
which admit the solution $\SF F\mapsto \lambda X.X\To X$, $\SF G, \SF H\mapsto \lambda X.X$ and $\alpha, \beta \mapsto \lambda X.X$.
Instead, checking if $\Gamma \vdash xy: \forall Z:Z$, where now $\Gamma(x)=\forall X.X\To X$ and $\Gamma(y)=Y$, yields the equations
\begin{align*}
\SF F X & = X \To X & \qquad
\SF G & = Y \\
\SF F (\alpha Z) & = \SF G   \To \SF H Z& \qquad
\SF H Z & = Z
\end{align*}
which have no solution (since one can deduce  $Z=\SF  HZ= Y$), showing that (TC) fails in this case (although $|\Gamma|\vdash xy: |\forall Z.Z|$ holds in the simply typed $\lambda$-calculus). 
%Observe that the dependency of $\alpha$ from $Z$ ensures the respect of the usual side condition of $\forall$I, for which the abstracted variable must be ``bindable'' in the context $\Gamma$: from the third and fourth equations it follows that, if one chooses $Z'$ not occurring free in $\Gamma$, one still has a solution by replacing $Z$ with $Z'$. 

%
% We show that any $\Nda$-unification problem translates  
%
%
%
%- Discuss type checking, first-order and second-order unification
%
%- Observe that typable terms are simply typable: \\
%Lemma. If $\Gamma \vdash t:A$ holds in Fat, then the first order unification for $|\Gamma|\vdash |t|:|A|$ is solvable.
%
%This suggests to first check typability by first-order unification. If it does not fail, then apply second-order unification.
%
%We exploit 
%- Levy's theorem: acyclic second-order unification is decidable.
%
%In particular, we show that TC can be reduced to a second-order unification problem which is of this form:
%- it can be associated with a first-order unification problem, so it is unsolvable if CYCLIC equations appear
%- otherwise, it reduces to a variant of acyclic second-order unification, for which a decision algorithm is provided.
%
%
%- Deduce Typability as usual
%

From the decidability of (TC) one can deduce the decidability of (T) by a standard argument: we can reduce (T) to (TC) by showing that a type $A$ such that $\Gamma \vdash \bb t:A$ holds exists iff $\Gamma  \vdash (\lambda xy.y)\bb t: \forall X. X\To X$ holds.
In fact, if $  \Gamma  \vdash \bb t:A$ holds in $\Nda$, then from $\Gamma \vdash \lambda xy.y: A\To \forall X.(X\To X)$ we deduce
$\Gamma  \vdash (\lambda xy.y)\bb t: \forall X.X\To X$. Conversely, from $\Gamma  \vdash (\lambda xy.y)\bb t: \forall X. X\To X$, we deduce that there exists a type $A$ such that 
%
%=\forall Y_{1}.\dots. \forall Y_{n}.B$, 
$\Gamma \vdash \lambda xy.y: A\To (X\To X)$ and 
$\Gamma  \vdash \bb t: A$ holds.
%, and
%%where for some variables $Z_{1},\dots, Z_{n}$, $C= B[Z_{1}/Y_{1},\dots,Z_{n}/Y_{n}]$; 
%since $x,y$ do not occur in $\bb t$, we can conclude that $\Gamma  \vdash \bb t:A$.
\begin{corollary}
(T) for Curry-style $\Nda$ is decidable.
\end{corollary}

\section{Equational Reasoning in System $\Nda$}\label{sec:equivalence0}
%\input{ceq}
% !TEX root = Undecidability.tex

As a consequence of Lemma \ref{lemma:stlc} from the previous section, all terms which are typable in Curry-style $\Nda$ are simply typable. In other words, $\Nda$ can be seen as a \emph{type refinement} system for $\STLC$, in the sense of \cite{Mellies2015}.
In particular, as we show below, the numerical functions which can be typed in $\Nda$ are precisely the simply typable ones (i.e.~the so-called \emph{extended polynomials} \cite{Schwichtenberg1975, Ferreira2018}).

For this reason, investigating the typable terms of $\Nda$ might seem not very interesting from a computational viewpoint. 
However, in this section we show that studying such terms can be interesting for equational reasoning. In fact, similarly to System F, standard notions of contextual equivalence are stronger than $\beta\eta$-equivalence, and one can exploit well-known techniques like \emph{free theorems} \cite{Wadler1989}, to compute equivalences of $\Nda$-typable terms (which do not hold when viewing these terms as typed in $\STLC$). 
%
%{equationally}. In standard type systems, beyond the standard notions of program equivalence arising from the operational semantics (i.e.~$\beta\eta$-equivalence $\simeq_{\beta\eta}$), there may exist several other \emph{congruences} arising from either denotational models or from some notion of \emph{contextual equivalence}.
%In STLC, it is well-known that $\simeq_{\beta\eta}$ coincides with the congruence induced by any infinite extensional model \cite{Statman1982}, as well as with several notions of contextual equivalence (see \cite{Barendregt2013}, \cite{Dosen2001}).
%

%
%
%
%
% \cite{}: two programs $t,u:A$ are equated when for all context $\TT C[\ ]:A \vdash O$ (where $O$ is a fixed type of \emph{observables}, $\TT C[t]\simeq_{\beta\eta}\TT C[u]:O$.  
%We will here focus on two standard variants of contextual equivalence, the one given by $O=\Bool=\forall X.X\To X\To X$, noted $\simeq_{\Bool}$, and the one given by $O=\Nat=\forall X.(X\To X)\To (X\To X)$, noted $\simeq_{\Nat}$.

We first recall two standard notions of contextual equivalence:

%
%In polymorphic type systems the picture is rather different, since $\beta\eta$-equivalence is usually weaker than the congruences arising from extensional models (see \cite{Bainbridge1990,Hasegawa2009}), and also weaker than standard the notions of contextual equivalence, defined as follows:
%
\begin{notation}
We let $\Bool=\forall X.X\To X\To X$ and $\Nat=\forall X.(X\To X)\To (X\To X)$. We let $\B t=\lambda xy.x$ and $\B f=\lambda xy.y$ indicate the two normal forms of type $\Bool$, and 
for all $n\in \mathbb N$, we let $\B n= \lambda fx.(f)^{n}x$ indicate the $n$-th Church numeral.
\end{notation}

\begin{definition}[contextual equivalence]
Let $F^{*}\in\{\Nda, \ML, \Fone, \mathrm F\}$.
For all closed terms $t,u$ of type $A$ in $F^{*}$, we let 
\begin{itemize}
\item $t\simeq_{\Bool}^{F^{*}}u:A$ iff for any context $\TT C[\ ]: A\vdash \Bool$ in $F^{*}$, $\TT C[t]\simeq_{\beta\eta}\TT C[u]$;\item $t\simeq_{\Nat}^{F*}u:A$ iff for any context $\TT C[\ ]: A\vdash \Nat$ in $F^{*}$, $\TT C[t]\simeq_{\beta\eta}\TT C[u]$.

\end{itemize}
\end{definition}
%contextual equivalence usually equates many more term than $\beta\eta$-equivalence. 

It is easily seen that $\simeq_{\Bool}^{F^{*}}$ and $\simeq_{\Nat}^{F^{*}}$ are congruences of the terms of $F^{*}$. Moreover, in System F these two congruences coincide, due to the fact that the identity relation $\TT{id}: \Nat \To \Nat \To \Bool$ is typable. 
Since this function is also typable in $\ML$, the same holds for $\ML$ and $\Fone$. 
On the other hand, since the identity relation is \emph{not} simply typable, we can deduce (see Lemma \ref{lemma:stlc1} below) that it is not typable in $\Nda$. For this reason the congruences  $\simeq_{\Bool}^{\Nda}$ and $\simeq_{\Nat}^{\Nda}$ must be treated separately. In what follows we will mostly focus on the latter, since the former identifies distinct normal forms of type $\Nat$, which is not convenient for obvious computational reasons.

\begin{remark}
The typability of the identity relation $\TT{id}$ implies that any extensional model must be \emph{infinite}, since for all $n\in \BB N$, the interpretations of $\B n$ and $\B{n+1}$ cannot coincide. 
Instead, it is not difficult to construct an extensional model of $\Nda$ in which any type is interpreted by a \emph{finite} set (to give an idea, let $\mathcal C_{k}$ be a collection of sets of cardinality bounded by a fixed $k\in \mathbb N$; one can let then $\model X\in \mathcal C_{k}$, $A\To B= \model B^{\model A}$ and 
$\model{\forall X.A}= \prod_{S\in \mathcal C_{k}}\model A[X\mapsto S]$). 
\end{remark}
%
%Several techniques have been studied to compute congruences stronger than $\beta\eta$-equivalence in System F, e.g.~\emph{free theorems} \cite{Wadler1989}, \emph{parametricity} \cite{Reynolds1983, Hermida2014}.

The so-called free theorems are a class of syntactic equations for typable terms which can be justified by relying on either relational parametricity \cite{Reynolds1983} or {dinaturality} \cite{Bainbridge1990}. We let $t\approx u:A$  indicate that $t,u$ have type $A$ in System F, and that the equivalence $t \simeq u$ can be deduced using $\beta$-, $\eta$-rules, standard congruence rules (i.e.~reflexivity, symmetry, transitivity and context closure), as well as instances of free theorems for System F. 

%the smallest congruence including $\beta\eta$-equivalences and free theorems.
%
%MORE PRECISE HERE.

Free theorems can be used to deduce contextual equivalence of $\Nda$-terms, thanks to the following:

\begin{lemma}[free theorems in $\Nda$]\label{lemma:freethm}
Let $t,u$ be terms of type $A$ in $\Nda$. If $t \approx u:A$, where $t,u$ are seen as terms of System F, then 
$t\simeq_{\Nat}^{\Nda}u:A$. 
\end{lemma}
\begin{proof}
From $t\approx u:A$ it follows $t\simeq_{\Nat}^{\mathrm F}u:A$, since $\simeq_{\Nat}^{\mathrm F}$ is the coarsest congruence not equating normal forms of type $\Nat$. From $t\simeq_{\Nat}^{\mathrm F}u:A$ we deduce 
$t\simeq_{\Nat}^{\Nda}u:A$, since any context in $\Nda$ is a context in F.
\end{proof}
 
% Put together, 
%Lemma \ref{lemma:stlc} and
%Lemma \ref{lemma:freethm} suggest to see $\Nda$ as a refinement of $\STLC$ enabling equational reasoning by free theorems. 

We discuss below two applications of free theorems to study (CE) in $\Nda$.

%
%
%- Observe that a corollary of undecidability is that the maximum theory of Fat does not coincide with the maximum theory of$\STLC$ 
%
%- However, they are interesting EQUATIONALLY, due to polymorphism. 
%In fact, in System F betaeta is weak with respect to CE. For this reason several techniques exist to compute equivalences. FREE THEOREMS arstlce a case, and admit an elegant categorical semantics (dinaturality). 
% 
%Fat =$\STLC$ +  Free theorems
%
%While betaeta coincides with STLC, contextual equivalence is coarser (there are less tests!). 
%
%- Define two variants of CE. Show that they coincide in F, but NOT in Fat.
%
%
%- Justify FREE THEOREMS: if t=u holds in F, then it holds in Fat. 
%Discuss FREE THEOREMS in F, Dinaturality ecc.
%

\subparagraph*{Categorical Products and Coproducts}

As mentioned in Section \ref{sec:pre}, the usual encoding of products and coproducts in System F preserves $\beta$-equivalence but not $\eta$-equivalence.
For this reason, the encodings of $\times$ and $+$ do not form categorical products and coproducts in System F up to $\beta\eta$-equivalence (more precisely, in the syntactic category in which objects are the types of System F and arrows are the typable terms up to $\simeq_{\beta\eta}$).  
 Instead, it is well-known \cite{Plotkin1993, Hasegawa2009, StudiaLogica} that $\eta$-equivalence of $\times$ and $+$ is preserved in System F up to free theorems: hence $\times$ and $+$ do form categorical products and coproducts in System F up to $\simeq_{\Nat}^{\mathrm F}$ (more precisely, in the syntactic category whose arrows are the typable terms up to $\simeq_{\Nat}^{\mathrm F}$).
 
In a similar way, the predicative encodings of $\times$ and $+$ in $\Nda$, although preserving some restricted case of $\eta$-equivalence, still do not form categorical products and coproducts in $\Nda$ up to $\simeq_{\beta\eta}$. We will show that they similarly do form categorical products and coproducts in $\Nda$ up to $\simeq_{\Nat}^{\Nda}$, as a consequence of the application of free theorems.

For simplicity, we here only consider the case of $+$. However, our argument scales straightforwardly to the encoding of all finite polynomial types, i.e.~of all types of the form $\sum_{i=1}^{k}\prod_{j=1}^{k_{i}}A_{ij}$ (see the pre-print \cite{SL2} for a more detailed discussion).

The fundamental step is showing that the impredicative and predicative encodings are equivalent up to free theorems:
\begin{lemma}\label{lemma:perreiras}
For all types $A,B,C$ and terms $x\mapsto A\vdash u:C $ and $x\mapsto B\vdash v: C$, the equivalence 
$ \IOor_{C}[y](\lambda x.u)(\lambda x.v)\approx \mathrm{Case}_{C}(y, x.u, x.v) :C$ holds in System $\mathrm F$.

\end{lemma}
\begin{proof}
The free theorem associated with the type $A\widetilde + B$ is the schematic equation
\begin{equation}\label{eq:ft1}
\mathrm{Case}_{E}(t_{1}, x.\TT C[t_{2}],  x.\TT C[t_{3}]) \approx  \TT C\Big[\mathrm{Case}_{D}( t_{1}, x.t_{2},  x.t_{3})\Big]
\end{equation}
where $t_{1}:A\widetilde + B$, $t_{2}[x:A]:D$, $t_{2}[x:B]:D$ and $\TT C[\ ]: D\vdash E$. In fact, this equation is an instance of the \emph{dinaturality} condition for the type $A\widetilde + B$ (see \cite{Plotkin1993, Hasegawa2009, Pistone2021}). 

We argue by induction on $C$:
\begin{itemize}
\item if $C=Y$, then $\IOor_{C}[y](\lambda x.u)(\lambda x.v)= yY(\lambda x.u)(\lambda x.v)=\mathrm{Case}_{C}(y,x.u,x.v)$;

\item if $C=C_{1}\To C_{2}$, then
\begin{align*}
\IOor_{C}[y](\lambda x.u)(\lambda x.v) & = 
\Big(
\lambda fgz.\IOor_{C_{2}}[y] (\lambda x.fxz)(\lambda x.gxz)\Big )
(\lambda x.u)(\lambda x.v) \\ 
& \stackrel{\text{[I.H.]}}{\approx}
\Big (\lambda fgz. \mathrm{Case}_{C_{2}}(y, x.fxz, x.gxz)\Big)
(\lambda x.u)(\lambda x.v) \\ 
& \simeq_{\beta}
\lambda z. \mathrm{Case}_{C_{2}}(y, x.uz, x.vz)\\
& \approx 
\lambda z. \Big(\mathrm{Case}_{C}(y, x.u, x.v)\Big)z
\\
& \simeq_{\eta}
\mathrm{Case}_{C}(y, x.u, x.v)
%
%
% \lambda z. \IOor_{C_{2}}[y] (\lambda a.(\lambda x.u)az)(\lambda b.(\lambda x.u)bz) \\
% & = \lambda z. \IOor_{C_{2}}[y] (\lambda x.uz)(\lambda x. uz) \\
\end{align*}
where in the penultimate step we applied Eq.~\eqref{eq:ft1} with the context $\TT C[\ ]= [\ ]z: C \vdash C_{2}$.

\item if $C=\forall Z.C'$, then
\begin{align*}
\IOor_{C}[y](\lambda x.u)(\lambda x.v) & = 
\Big(
\lambda fg.\Lambda Z.\IOor_{C'}[y] (\lambda x.fxZ)(\lambda x.gxZ)\Big )
(\lambda x.u)(\lambda x.v) \\ 
& \stackrel{\text{[I.H.]}}{\approx}
\Big (\lambda fg.\Lambda Z. \mathrm{Case}_{C'}(y, x.fxZ, x.gxZ)\Big)
(\lambda x.u)(\lambda x.v) \\ 
& \simeq_{\beta}
\Lambda Z. \mathrm{Case}_{C'}(y, x.uZ, x.vZ)\\
& \approx 
\Lambda Z. \Big(\mathrm{Case}_{C}(y, x.u, x.v)\Big)Z
\\
& \simeq_{\eta}
\mathrm{Case}_{C}(y, x.u, x.v)
%
%
% \lambda z. \IOor_{C_{2}}[y] (\lambda a.(\lambda x.u)az)(\lambda b.(\lambda x.u)bz) \\
% & = \lambda z. \IOor_{C_{2}}[y] (\lambda x.uz)(\lambda x. uz) \\
\end{align*}
where in the penultimate step we applied Eq.~\eqref{eq:ft1} with the context $\TT C[\ ]= [\ ]Z: C \vdash C'$.
\end{itemize}
\end{proof}

\begin{proposition}
$A\widetilde + B$ is a categorical coproduct in $\Nda$ up to $\simeq_{\Nat}^{\Nda}$.
\end{proposition}
\begin{proof}
It suffices to check that the $\eta$-rule of the coproduct 
(see \cite{LambekScott}) holds in $\Nda$. By translating this rule in F one obtains the equation
$$
y \approx \mathrm{Case}_{A\widetilde + B}(y, x.\iota_{1}(x), x.\iota_{2}(x)) : A\widetilde + B
$$
which holds in F up to free theorems (see \cite{Plotkin1993, Hasegawa2009, StudiaLogica}). 
Using Lemma \ref{lemma:perreiras} we thus deduce that $y \approx \IOor_{A\widetilde + B}[y](\lambda x.\iota_{1}(x))(\lambda x.\iota_{2}(x)) : A\widetilde + B
$ holds in F, and by Lemma \ref{lemma:freethm} we deduce 
 $y \simeq_{\Nat}^{\Nda} \IOor_{A\widetilde + B}[y](\lambda x.\iota_{1}(x))(\lambda x.\iota_{2}(x)) : A\widetilde + B
$.
\end{proof}

%- (1) A natural application is Finite Initial Algebras:\\
%		- illustrate just the case of +. \\
%		- In F you have a coproduct (using free theorems). The encoding in Fat is equivalent to the one in F up to free theorems. So you have a coproduct in Fat.

\subparagraph*{Numerical Functions}

We now consider the representable numerical functions, that is, the typable terms of type 
$\Nat \To \dots \To \Nat \To \Nat$. In this case we can strengthen Lemma \ref{lemma:stlc}:
\begin{lemma}\label{lemma:stlc1}
For any $\beta$-normal $\lambda$-term $t$, 
$\vdash t:\Nat \To \dots \To \Nat \To \Nat$ holds in Curry-style $\Nda$ iff 
$\vdash t:|\Nat| \To \dots \To| \Nat |\To |\Nat| $ holds in $\STLC$.
\end{lemma}
\begin{proof}
One direction follows from Lemma \ref{lemma:stlc}. For the converse one, let $t$ (which we can suppose w.l.o.g.~to be of the form $\lambda x_{1}\dots x_{k}.u$) be such that $\vdash t:|\Nat| \To \dots \To| \Nat |\To |\Nat| $. By letting $\Nat[X]=(X\To X)\To (X\To X)$ we deduce that 
$\{x_{i}\mapsto \Nat[X]\}\vdash u: \Nat[X]$ holds in $\Nda$, and thus that
$\{x_{i}\mapsto \Nat\}\vdash u: \Nat[X]$ holds too, from which we conclude
$\vdash u:\Nat \To \dots \To \Nat \To \Nat$.
\end{proof}

A consequence of Lemma \ref{lemma:stlc1} is that the representable numerical functions in $\Nda$ are precisely the extended polynomials, i.e.~the smallest class of functions arising from projections, constant functions, addition, multiplication and the $\TT{iszero}$ function. 
Instead, it is well-known that the predecessor function (which is not an extended polynomial) is
typable in $\ML$ \cite{Leivant1983b} and, more generally, the representable functions of $\ML$ 
are included in the class $\C E_{3}$ of the Grzegorczyk hierarchy \cite{Leivant1991}.%, but, to the best of the authors' knowledge, the exact characterization of this class is still an open problem.

Still, in both $\STLC$ and $\Nda$ the same extended polynomial can be represented by different normal forms.  For instance the two normal forms $\lambda xy fz. x (yf)z $ and $\lambda xyfz. y(xf)z$ (encoding the algorithms $n,m \mapsto 
\underbrace{m+\dots+m}_{n\text{ times}}$ and $n,m \mapsto 
\underbrace{n+\dots+n}_{m\text{ times}}$) both represent the multiplication function.

In System F, one can show that all primitive recursive functions are \emph{uniquely defined} up to free theorems, that is, that for any two terms $t,u$ representing the same primitive recursive function, one can prove $t\approx u$ (see \cite{Pistone2020b}, Section 7.5). Using Lemma \ref{lemma:freethm} we deduce then:

\begin{lemma}\label{lemma:eqreas}
For all $t,u: \Nat\To \dots \To \Nat \To \Nat$ in $\mathrm F^{*}\in\{\Nda, \ML,\Fone\}$, if for all $p_{1},\dots, p_{k}\in \BB N$, 
$t\B p_{1}\dots \B p_{k}\simeq_{\beta\eta} u\B p_{1}\dots \B p_{k}:\Nat$, then $t\simeq_{\Nat}^{\mathrm F^{*}}u$.
\end{lemma}

The problem $\mathrm{Eq}_{\C C}$ of deciding $f=g$, where $f,g$ belong to some subclass $\C C$ of the primitive recursive functions, is well-investigated. In particular, it is known that:
\begin{itemize}
\item if $\C C$ is the class of extended polynomials, then $\mathrm{Eq}_{\C C}$ is decidable \cite{Zaionc1997};
\item if $\C C$ contains projections, constants, +, $\times$ and bounded multiplication, then $\mathrm{Eq}_{\C C}$ is undecidable \cite{Lee1969}.
%\item if $\C C$ contains projections, constants, $+$, $\times$ and $\dotdiv$, then $\mathrm{Eq}_{\C C}$ is undecidable \cite{Lee1969, Okada1999}.
\end{itemize}

From these facts, using Lemma \ref{lemma:eqreas}, we deduce then:
\begin{proposition}
\begin{itemize}
\item[i.] The problem of deciding $\simeq_{\Nat}^{\Nda}$ over numerical functions in $\Nda$ is decidable.
\item[ii.] The problem of deciding $\simeq_{\Nat}^{\mathrm F^{*}}$ over numerical functions in $\mathrm F^{*}\in\{\ML, \Fone\}$ is undecidable.
\end{itemize}
\end{proposition}
\begin{proof}
i. is immediate from Lemma \ref{lemma:stlc1} and Lemma \ref{lemma:eqreas}. To prove ii. it suffices to show that the representable functions in $\ML$ are closed under bounded multiplication. We show this fact in detail in App.~\ref{appD}. 
\end{proof}

An immediate corollary is that (CE) is undecidable in both $\ML$ and $\Fone$.

%
% \emph{equational reasoning} on numerical functions. In other words, one can 
%
%
%
%- (2) Decidability of equational reasoning:\\
%		- justify equational reasoning: e.g. x+y=y+x etc. Use (CE) and mention works on dinaturality (Roman Pare and us). This allows to draw another distinction between $\Nda$, $\ML$ and $\Fone$ 
%		
%		
%		- Observe that (1) +,$\times$ and $-$ are undecidable (Scott Okada), (2) +,$\times$ and $\Pi$ are undecidable (Tizio), (3) extended polynomials are decidable (Zaionc). 
%		
%		
%				- Zaionc's decidability theorem\\
%		- ML-representable functions and skew representable functions. Undecidability theorem (bounded multiplication is ML-representable).
% (i.e. the closed terms of type $\Nat\To \dots \To \Nat \To \Nat$, where $\Nat=\forall X.(X\To X)\To (X\To X)$) 

\section{Contextual Equivalence is Undecidable}\label{sec:equivalence}

% !TEX root = Undecidability.tex

In this section we show that the congruences $\simeq_{\Nat}^{\Nda}$ and $\simeq_{\Bool}^{\Nda}$ are both undecidable. 
To do this, we will reduce the type inhabitation problem for a suitable extension of $\Nda$ to contextual equivalence.
We discuss in some detail the undecidability argument for $\simeq_{\Bool}^{\Nda}$, and we postpone the  similar argument for $\simeq_{\Nat}^{\Nda}$ to App.~\ref{appE}.

Let $\Ndac$ be $\Nda$ extended with a type constant $\clubsuit$ and a term constant $\star:\clubsuit$.
It is not difficult to see that the undecidability argument for (TI) from Section 3 also applies to $\Ndac$. 

Let ${\widetilde \top}:\forall X.X\To X$ and 
$\bb{\mathsf{Id}}:= \bb{\Lambda X.\lambda x.x}$ be the unique closed $\beta$-normal term of type $\widetilde \top$.
%Moreover, for each type $A$ of $\Ndac$, we let $A^{*}= Y\To A[Y/\clubsuit]$.
%Moreover, we let $\Bool$ indicate the type $\forall X.X\To X\to X$ and $\bb{\B t}:= \bb{\Lambda X.\lambda x_{1}x_{2}.x_{1}}$ and $\bb{\B f}:= \bb{\Lambda X.\lambda x_{1}x_{2}.x_{2}}$ indicate the unique two closed $\beta$-normal terms of type $\Bool$.
%In App.~\ref{appB} it is shown that (TI) for $\Ndac$ is equivalent to (TI) for $\Nda$, and thus, undecidable.

The fundamental idea will be to construct, for each type $A$ of $\Ndac$, two terms $t_{A},u_{A}$ of type 
$(A^{*}\widetilde  + \widetilde \top)\To \Bool$ (where $A^{*}=Y\To A[Y/\clubsuit]$, for some fresh $Y$), such that $t_{A}\simeq_{\Bool}^{\Nda }u_{A}$ holds in $\Nda$ iff 
$A$ is inhabited in $\Ndac$.

Let us fix a type $A$ of $\Ndac$,
a variable $Y$ not occurring free in $A$, and let $A^{*}= Y\To A[Y/\clubsuit]$. We let $\bb{u_{A}}, \bb{v_{A}}$ be the terms below:
\begin{align*}
\bb{u_{A}}    = \bb{ \lambda x. \B{f} } \qquad 
\bb{v_{A} }   = \lambda x.  \IOor_{\Bool}[\ ] (  \lambda x. \B t) ( \lambda x. \B f)
\end{align*}

%In the following we will indicate as $\TT C[\ ]$ a term with a designated free variable $x$; for any term $u$, we let $\TT C[u]$ be a shorthand for $\TT C[u/x]$. 
In the following, for a term context $\TT K[\ ]$, we let $\TT K[\ ]: A \vdash^{\Gamma} B$ be a shorthand for $\Gamma, x\mapsto A\vdash\TT K[\ ]: B$. % and $\TT C[\ ]: A \vdash B$ be a shorthand for $\TT C[\ ]: A \vdash^{\emptyset} B$.

We let $\BB G_{1}$-$\BB G_{4}$ be the families of term contexts defined by mutual recursion as shown in Fig.~\ref{fig:contexts}.
One can check that these contexts cannot separate $u_{A}$ and $v_{A}$:

\begin{figure}
\fbox{
\begin{minipage}{0.88\textwidth}
\vskip-0.3cm
\begin{align*}
\BB G_{1}: \qquad \rr{\TT C}[\ ] &:: =   \bb{x_{i}} \mid \rr{\TT E}[\ ] Y\bb{\TT C[\ ]\TT C[\ ]} \\
\BB G_{2}:\qquad \rr{\TT D}[\ ] & ::=  \bb{ z}(\Lambda W.\lambda w.\rr{\TT F}[\ ])\mid \bb{\TT E}[\ ]\bb{W} \TT D[\ ] \TT D[\ ]\\
\BB G_{3}: \qquad \rr{\TT E}[\ ]& :: =   \B t\mid \B f\mid 
%\Lambda Z.\lambda x.\lambda x'. ([\ ] \Lambda Y.\lambda y.\lambda z.\rr{\TT D}[\ ])\TT C_{1}[\ ]\TT C_{2}[\ ] 
x (\Lambda Y.\lambda y.\lambda z.\rr{\TT D}[\ ])
\\
\BB G_{4}: \qquad \rr{\TT F}[\ ]&::= w\mid \TT E[\ ]W\TT F[\ ]\TT F[\ ]
%\bb{ \mathsf{Id}} \mid\bb{ \Lambda Z.\lambda z.\TT E}[ \  ]\bb{Zzz}
\end{align*}
%\end{minipage}
%\begin{minipage}{0.3\textwidth}
%\begin{equation*}
%\begin{split}
%&  (\rr{\TT E}[\ ]\in \BB G_{3})\\
% & (\rr{\TT E}[\ ]\in \BB G_{3}, \rr{\TT F}\in \BB G_{4}) \\
%& (\rr{\TT D}[\ ]\in \BB G_{2}
%%, \TT C_{i}[\ ]\in \BB G_{1}
%)\\
%& (\rr{\TT E}[\ ]\in \BB G_{3})
%\end{split}
%\end{equation*}
\end{minipage}
}
\caption{Contexts $\BB G_{1}$-$\BB G_{4}$.}
\label{fig:contexts}
\end{figure}

%\begin{figure}%, depending on two free variables $x_1,x_{2}$:
%\caption{Sets of contexts $\BB G_{1}-\BB G_{4}$.}
%\label{fig:contexts}
%\end{figure}

\begin{lemma}\label{lemma2}
\begin{enumerate}
\item For all $\rr{\TT C}[\ ]\in \BB G_1$, $\rr{\TT C}[\bb{u_{A}}]\simeq_{\beta\eta}\rr{\TT C}[\bb{v_{A}}]$.
\item for all $\rr{\TT D}[ \ ]\in \BB G_2$, $\rr{\TT D}[\bb{u_{A}}]\simeq_{\beta\eta}\rr{\TT D}[\bb{v_{A}}]\simeq_{\beta\eta} \bb{z_{i} \mathsf{Id}}$.
\item for all $\rr{\TT E}[ \ ]\in \BB G_3$, $\rr{\TT E}[\bb{u_{A}}]\simeq_{\beta\eta}\rr{\TT E}[\bb{v_{A}}]$.
%\bb{\B f}$.
\item for all $\rr{\TT F}[ \ ]\in \BB G_4$, $\rr{\TT F}[\bb{u_{A}}]\simeq_{\beta\eta}\rr{\TT F}[\bb{v_{A}}]\simeq_{\beta\eta} \bb{w_{i}}$.

\end{enumerate}
%
%For all $\rr{\TT C}[ \ ]\in \BB G_1$, $\rr{\TT C}[\bb{u_{A}}]\simeq_{\beta\eta}\rr{\TT C}[\bb{v_{A}}]$.
\end{lemma}

The key ingredient is a lemma stating that the families of contexts $\BB G_{1}$-$\BB G_{4}$ are exhaustive precisely when $A$ is not inhabited in $\Ndac$:

\begin{lemma}\label{lemma1}
Let $\rr{\TT K}[\ ]: (A^{*}\widetilde +\widetilde \top)\To \Bool \vdash ^{x\mapsto Z,x'\mapsto Z} Z$ be a $\beta$-normal term context. Then either $A$ is inhabited in $\Ndac$ or $\rr{\TT K}[\ ]\in\BB G_1$.
\end{lemma}
\begin{proof}
%See App.~\ref{appE}.
%\end{proof}
%
%
%\begin{proof}[Proof of Lemma \ref{lemma1}]
We will prove the following claim: either there exists contexts $\Gamma,\Theta,\Delta, \Sigma$, where
\begin{align*}\label{eq:delta0}
\Gamma& = \{ x_{1} \mapsto Z_{1},x'_{1}\mapsto Z_{1},\dots, x_{p}\mapsto Z_{p},x'_{p}\mapsto Z_{p}\}&   
\Theta & = \{ w_{1}\mapsto W_{1},\dots, w_{q}\mapsto W_{q}\}\\
 \Delta& =\{ y_{1}\mapsto A^{*}\To Y_{1},\dots, y_{r}\mapsto A^{*}\To Y_{r}\} & 
 \Sigma & =\{ z_{1}\mapsto\widetilde \top \To Y_{1},\dots, z_{r}\mapsto  \widetilde \top\To Y_{r}\}
\end{align*}
for some $p,q,r\in \BB N$ and variables $Z_{1},\dots, Z_{p},W_{1},\dots, W_{q},Y_{1},\dots, Y_{r}$  pairwise distinct and disjoint from $A$, 
and a context $\rr{\TT H}[\ ]:  (A^{*}\widetilde +\widetilde \top)\To\Bool \vdash^{\Gamma,\Theta,\Delta, \Sigma} A^{*}$, or $\rr{\TT K}[\ ]\in \BB G_1$.
If the main claim is true we can deduce the statement of the lemma as follows: suppose $\rr{\TT K}[\ ]\notin \BB G_1$.
Then let $\theta$ be the substitution sending all variables in $\Gamma,\Theta,\Delta,\Sigma$ plus $Y$ onto $\clubsuit$ and being the identity on all other variables.
Then $\rr{\TT H\theta}[\ ]: ( (\clubsuit \To A)\widetilde + \widetilde \top)\To\Bool \vdash^{\Gamma\theta, \Theta\theta, \Delta\theta, \Sigma\theta}:  \clubsuit \To A$. Then we have $\Gamma\theta,\Theta\theta,\Delta\theta,\Sigma\theta\vdash \bb t:A$, where 
$
\bb t= \bb{\TT H\theta}[ \lambda x. \B t ]\bb{\star}
$ and we can conclude that $\vdash \bb{t'}:A$ holds where $\bb{t'}$ is obtained from $\bb t$ by 
substituting the variables in $\Gamma$ and $\Theta$ by $\star$ and those in $\Delta$ and $\Sigma$ by $\lambda x.\star$.

Let us prove the main claim.
Suppose by contradiction that for no $\Gamma,\Theta,\Delta, \Sigma$ there exists a context $\rr{\TT H}[\ ]: (A^{*}\widetilde +\widetilde \top)\To\Bool \vdash^{\Gamma,\Theta,\Delta,\Sigma} A^{*}$. We will show by simultaneous induction the following claims:
\begin{enumerate}
\item for all $\Gamma,\Theta,\Delta, \Sigma$ as above, if $\rr{\TT K}[\ ]: (A^{*}\widetilde +\widetilde \top)\To\Bool \vdash^{\Gamma,\Theta,\Delta, \Sigma} Z_{i}$, then $\rr{\TT K}[\ ]\in \BB G_1$;
\item for all $\Gamma,\Theta,\Delta, \Sigma$ as above, if $\rr{\TT K}[\ ]: (A^{*}\widetilde +\widetilde \top)\To\Bool \vdash^{\Gamma,\Theta,\Delta, \Sigma} Y_{i}$, then $\rr{\TT K}[\ ]\in \BB G_2$;

\item for all $\Gamma,\Theta,\Delta, \Sigma$ as above, if $\rr{\TT K}[\ ]: (A^{*}\widetilde +\widetilde \top)\To\Bool \vdash^{\Gamma,\Theta,\Delta, \Sigma} \Bool$ and $\rr{\TT K}[\ ]$ is an elimination context, then $\rr{\TT K}[\ ]\in \BB G_3$;

\item for all $\Gamma,\Theta,\Delta, \Sigma$ as above, if $\rr{\TT K}[\ ]: (A^{*}\widetilde +\widetilde \top)\To\Bool \vdash^{\Gamma,\Theta,\Delta, \Sigma} W_{i}$, then $\rr{\TT K}[\ ]\in \BB G_4$.

\end{enumerate}
The main claim then follows from 1. by taking $\Gamma=\{x\mapsto Z,x'\mapsto Z\} $ and $\Theta=\Delta=\Sigma=\emptyset$.

We argue for each case separately:% (and we suppose w.l.o.g. that the contexts are $\beta$-normal):
\begin{enumerate}
\item
There exist two possibilities for $\rr{\TT K}[\ ]$:
\begin{enumerate}
\item $\rr{\TT K}[\ ]=\bb{x_{i}}$, hence $\rr{\TT K}[\ ]\in \BB G_1$;

\item $\rr{\TT K}[\ ]=  \rr{\TT K'}[\ ] \bb{Z \TT K_{1}[\ ]\TT K_{2}[\ ]}$, where $\rr{\TT K'}[\ ]: (A^{*}\widetilde +\widetilde \top)\To\Bool \vdash^{\Gamma,\Theta,\Delta,\Sigma} \Bool$ and 
 $\rr{\TT K_{i}}[\ ]: (A^{*}\widetilde +\widetilde \top)\To\Bool \vdash^{\Gamma,\Theta,\Delta,\Sigma} Z$, and where $\TT K'[\ ]$ is an elimination context. By the induction hypothesis then $\rr{\TT K'}[\ ]\in \BB G_3, \TT K_{i}[\ ]\in \BB G_{1}$, hence $\rr{\TT K}[\ ]\in \BB G_1$.

\end{enumerate}

\item There exist three possibilities for $\rr{\TT D}[\ ]$: %Let $W,y,z$ be fresh variables and $\Delta'=\Delta\cup \{y: A^{*}\To W \}$, $\Sigma'=\Sigma\cup\{z: \widetilde \top\To W\}$.
\begin{enumerate}
\item $\rr{\TT K}[\ ]=\bb{y_{i}\rr{\TT K'}}[\ ]$,
where $\rr{\TT K'}[\ ]: (A^{*}\widetilde +\widetilde \top)\To\Bool \vdash^{\Gamma,\Theta,\Delta,\Sigma} A^{*}$, but this case is excluded by the hypothesis;

\item $\rr{\TT K}[\ ]=z_{i} (\Lambda W.\lambda w.\rr{\TT K'}[\ ])$, where $\rr{ \TT K'}[\ ]:  (A^{*}\widetilde +\widetilde \top)\To\Bool \vdash^{\Gamma,\Theta\cup \{w\mapsto W\}, \Delta,\Sigma} W$ and where $W$ does not occur in $\Gamma,\Theta,\Delta,\Sigma$. By the induction hypothesis then 
$\rr{\TT K'}[\ ]\in \BB G_4$, hence $\rr{\TT K}[\ ]\in \BB G_2$;

\item $\rr{\TT K}[\ ]=\rr{\TT K'} [\  ]   Y_{i} \rr{\TT K_{1}}[\ ] \rr{\TT K_{2}}[\ ]    $, where 
$\rr{\TT K'}[\ ]   : (A^{*}\widetilde +\widetilde \top)\To\Bool \vdash^{\Gamma,\Theta, \Delta,\Sigma} \Bool$, 
$\rr{\TT K_{i}}[\ ] :(A^{*}\widetilde +\widetilde \top)\To\Bool \vdash^{\Gamma,\Theta, \Delta,\Sigma}  Y_{i}$, and $\rr{\TT K'}[\ ]$ is an elimination context.
 By the induction hypothesis this implies $\rr{\TT K'}[\ ]\in \BB G_3$ and $\rr{\TT K_{i}}\in \BB G_2$, so we can conclude $\rr{\TT K}[\ ]\in \BB G_2$.

%\item[2.4] $\rr{\TT D}[\ ]=\bb{\Lambda W.\lambda yz. y} \Big (\rr{\TT G}\big[ \rr{\TT E}[\ ] /\BB G_2][ \ ] \Big)  $, for some contexts 
%\begin{equation*}
%\begin{split}
%\rr{\TT E}[\ ]  & : (A^{*}\widetilde +\widetilde \top)\To \Bool \vdash^{\Gamma, \Delta',\Sigma'} \Bool\\
%\rr{\TT G}[\ ] &:(A^{*}\widetilde +\widetilde \top)\To \Bool \vdash^{\Gamma, \Delta',\Sigma', \BB G_2:\Bool}  A^{*}
%\end{split}
%\end{equation*}
%This  case is excluded by the hypothesis, as by substituting a Boolean for $\BB G_2$ in $\rr{\TT G}[\ ]$ one obtains a context which contradicts the hypothesis.
%
%
%
%\item[2.5] $\rr{\TT D}[\ ]=\bb{\Lambda W.\lambda yz. z} \rr{\TT F}[\ ] $, for some contexts 
%\begin{equation*}
%\begin{split}
%\rr{\TT E}[\ ]  & : (A^{*}\widetilde +\widetilde \top)\To \Bool \vdash^{\Gamma, \Delta',\Sigma'} \Bool\\
%\rr{\TT F}[\ ] &:(A^{*}\widetilde +\widetilde \top)\To \Bool \vdash^{\Gamma, \Delta',\Sigma', \BB G_2:\Bool}  \widetilde \top
%\end{split}
%\end{equation*}
%By the induction hypothesis then $\rr{\TT F}[\ ]\in \BB G_4$. Hence either $\rr{\TT F}[\ ]= \mathsf{Id}$, which implies $\rr{\TT D}[\ ]\in \BB G_2$, or $\rr{\TT F}[\ ]=\bb{ \Lambda Z.\lambda z.}\rr{\TT E}[\ ]\bb{Zzz}$. In this case $\rr{\TT E}[\ ]$ must be an elimination context, hence by the induction hypothesis, $\rr{\TT E}[\ ]\in \BB G_3$, whence $\rr{\TT D}[\ ]\in\BB G_2$.

\end{enumerate}
\item If $\TT K[\ ]$ is an elimination context, then it must be $\rr{\TT K}[\ ]= x\rr{\TT K'}[\ ]$, where $\rr{\TT K'}[\ ]:
(A^{*}\widetilde +\widetilde \top)\To\Bool \vdash^{\Gamma\cup\{x_{1}\mapsto Z',x_{2}\mapsto Z''\},\Theta, \Delta,\Sigma} A^{*}\widetilde +\widetilde \top$.
%
%$\rr{\TT C_{i}}[\ ]:
%(A^{*}\widetilde +\widetilde \top)\To\Bool \vdash^{\Gamma\cup\{x_{1}\mapsto Z',x_{2}\mapsto Z''\},\Theta, \Delta,\Sigma} Z'$, 
% and $Z'$ does not occur in either of $\Gamma,\Theta,\Sigma,\Delta$;   
Moreover, $\TT K'$ must be of the form $\Lambda Y.\lambda y.\lambda z.\TT K''[\ ]$, where 
 $\rr{\TT K''}[\ ]:
(A^{*}\widetilde +\widetilde \top)\To\Bool \vdash^{\Gamma\cup\{x_{1}\mapsto Z',x_{2}\mapsto Z''\},\Theta, \Delta\cup \{y\mapsto A^{*}\To Y\},\Sigma\cup \{z\mapsto \widetilde \top\To Y\}} Y$, and where $Y$ is distinct from all variables in $\Gamma\cup\{x_{1}\mapsto Z',x_{2}\mapsto Z''\},\Theta,\Delta, \Sigma$; then by the induction hypothesis we deduce 
$\rr{\TT K''}[\ ]\in \BB G_2$, and thus $\rr{\TT K}[\ ]\in \BB G_3$.
%
%\item If $\TT E[\ ]\neq \B t, \B f$, then it must be $\rr{\TT E}[\ ]= \Lambda Z'.\lambda x'_{1}.\lambda x'_{2}.
%x\rr{\TT D}[\ ]\TT C_{1}[\ ]\TT C_{2}[\ ]$, where $\rr{\TT D}[\ ]:
%(A^{*}\widetilde +\widetilde \top)\To\Bool \vdash^{\Gamma\cup\{x_{1}\mapsto Z',x_{2}\mapsto Z''\},\Theta, \Delta,\Sigma} A^{*}\widetilde +\widetilde \top$,
%$\rr{\TT C_{i}}[\ ]:
%(A^{*}\widetilde +\widetilde \top)\To\Bool \vdash^{\Gamma\cup\{x_{1}\mapsto Z',x_{2}\mapsto Z''\},\Theta, \Delta,\Sigma} Z'$, 
% and $Z'$ does not occur in either of $\Gamma,\Theta,\Sigma,\Delta$;   
%moreover, $\TT D$ must be of the form $\Lambda Y.\lambda y.\lambda z.\TT D'[\ ]$, where 
% $\rr{\TT D}'[\ ]:
%(A^{*}\widetilde +\widetilde \top)\To\Bool \vdash^{\Gamma\cup\{x_{1}\mapsto Z',x_{2}\mapsto Z''\},\Theta, \Delta\cup \{y\mapsto A^{*}\To Y\},\Sigma\cup \{z\mapsto \widetilde \top\To Y\}} A^{*}\widetilde +\widetilde \top$, and where $Y$ is distinct from all variables in $\Gamma\cup\{x_{1}\mapsto Z',x_{2}\mapsto Z''\},\Theta,\Delta, \Sigma$; then by the induction hypothesis we deduce 
%$\rr{\TT D}'[\ ]\in \BB G_2$ and $\TT C_{i}\in \BB G_{1}$, from which we deduce $\rr{\TT E}[\ ]\in \BB G_3$.

\item There are two possible cases:
\begin{enumerate}
\item $\rr{\TT K}[\ ]=w_{i}$, hence $\rr{\TT K}[\ ]\in \BB G_4$;

\item $\rr{\TT K}[\ ]= \rr{\TT K'}[\ ]\bb{W_{i}\TT K_{1}\TT K_{2}}$, where $\rr{\TT K'}[\ ]: (A^{*}\widetilde +\widetilde \top)\To \Bool \vdash^{\Gamma,\Theta, \Delta,\Sigma} \Bool$, $\rr{\TT K_{i}}[\ ]: (A^{*}\widetilde +\widetilde \top)\To \Bool \vdash^{\Gamma,\Theta, \Delta,\Sigma} W_{i}$ and $\TT K'[\ ]$ is an elimination context. By the induction hypothesis this implies $\rr{\TT K'}[\ ]\in \BB G_3$ and $\TT K_{i}\in \BB G_{4}$, whence $\rr{\TT K}[\ ]\in \BB G_4$.
\end{enumerate}
\end{enumerate}
\end{proof}

\begin{proposition}\label{prop:sbo}
$\bb{u_{A}}\not\simeq_{\Bool}^{\Nda}\bb{v_{A}}$ iff $A$ is inhabited in $\Ndac$.
\end{proposition}
\begin{proof}
Suppose $\vdash t:A$ holds in $\Ndac$, then by letting $\bb{t^{*}}= \bb{\lambda y. t}[\bb y/\bb{\star}]$ we deduce
$\vdash \bb{t^{*}}: A^{*}$ and we can show $\bb{u_{A}}\not\simeq_{\Bool}^{\Nda}\bb{v_{A}}$, by letting
$
\rr{\TT K}[\ ] =  x ( \bb{\iota_{1}(t^{*})})
$, 
since $\rr{\TT K}[ \bb{u_{A}}]\simeq_{\beta}\bb{\B f} $ and
$\rr{\TT K}[\bb{v_{A}}]\simeq_{\beta}
\bb{\IOor_{\Bool}[\iota_{1}(t^{*})]( \lambda x.\B t)( \lambda x.\B f)} \simeq_{\beta\eta}
\bb{(\lambda x.\B t) t^{*}} \simeq_{\beta} \bb{\B t}$.

Conversely, suppose $A$ is not inhabited in $\Ndac$.
Any context $\rr{\TT K}[\ ]: (A^{*}\widetilde +\widetilde \top)\To \Bool\vdash ^{\emptyset} \Bool$ can be written, up to $\eta$-equivalence, as $\rr{\TT K}[\ ]=\bb{ \Lambda Z.\lambda x_1x_{2}.}\rr{\TT K'}[\ ]$, with
 $\rr{\TT K'}[\ ]: (A^{*}\widetilde +\widetilde \top)\To \Bool\vdash ^{\bb{x_{1}}\mapsto Z, \bb{x_{2}}\mapsto Z} \Bool$. As we can suppose $\rr{\TT K}[\ ]$ to be $\beta$-normal, by Lemma \ref{lemma1}, it must be $\rr{\TT K'}[\ ]\in \BB G_{1}$. Hence, by Lemma \ref{lemma2} we deduce that $\rr{\TT K}[\bb{u_{A}}]\simeq_{\beta\eta} \rr{\TT K}[\bb{v_{A}}]$.
%  Moreover, as a closed $\beta$-normal term of type $\Bool$ must be either $\B t$ or $\B f$, from $\rr{\TT C}[\bb{u_{A}}]\simeq_{\beta\eta} \rr{\TT C}[\bb{v_{A}}]$ we can deduce
% $\rr{\TT C}[\bb{u_{A}}]\simeq_{\beta} \rr{\TT C}[\bb{v_{A}}]$.
%We have thus shown that $\bb{u_{A}}\simeq_{\Bool}^{\Nda}\bb{v_{A}}$.
\end{proof}

\begin{theorem}
The congruences $\simeq_{\Bool}^{\Nda}$ and $\simeq_{\Nat}^{\Nda}$ are both undecidable.
%The $\mathrm{(CE)}$ problem for any of the systems $\Nda_{\alpha,\beta}$ is undecidable.
\end{theorem}
%\begin{proof}
%By Proposition \ref{prop:sbo}, if $\simeq_{\Bool}^{\Nda}$ were decidable, we would be able to decide type inhabitation in $\Ndac$, which is impossible.\end{proof}

%%
%\section{$\Nda$ as a dependently typed system}\label{sec:dep}
%\input{dependent}

\section{Conclusion}
% !TEX root = Undecidability.tex

%In this paper we proved the decidability of type checking and typability in System $\Nda$, and 
%the undecidability of contextual equivalence in $\Nda$, $\ML$ and all systems of stratified polymorphism $\mathrm F_{n}$, for $n>0$. 

\subparagraph*{Related works}

The literature on ML-polymorphism, both at theoretical and applicative level, is vast. Several extensions of ML to account for first-class polymorphism while retaining a decidable type-checking have been investigated, mostly following two directions: first, that of considering type systems with explicit type annotations (as the system $\mathsf{PolyML}$ \cite{Garrigue1997}); second, that of encoding first-class polymorphism in a ML-style system by means of \emph{coercions} (as in System $\mathrm{Fc}$ \cite{Sulzmann2007} or in in $\ML^{\mathrm F}$ \cite{Remy2003}). 
In the last case, coherently with our discussion on FOU and SOU, the price to pay to remain decidable is that self-applications of $\lambda$-abstracted variables must come with explicit type annotations.
This approach is currently followed in the design of the Haskell compiler, which supports first-class polymorphism.

Predicative restrictions of System F and their expressive power have been also largely investigated \cite{Leivant1989, Leivant1991, Danner99}. For example, the numerical functions representable in Leivant's finitely stratified polymorphism are precisely those at the 3rd level of Grzegorczyk's hierarchy \cite{Leivant1991}, and transfinitely stratified systems have been shown to represent all primitive recursive functions \cite{Danner99}. In \cite{Leivant1994} a system with expressive power comparable to System $\Nda$ is shown to characterize the polytime functions.

Research by Ferreira and her collaborators on System $\Nda$ has mostly focused on predicative translations of intuitionistic logic and their reduction properties \cite{Ferreira2013, Ferreira2009, Ferreira2020}. As mentioned before, these translations rely on the observation that for certain types the unrestricted $\forall$E-rule is admissible in $\Nda$. The characterization of the class of types satisfying this property is an open problem (a partial characterization is described in \cite{PistoneInsta}).

Another way to obtain interesting subsystems of System F is by restricting the class of types which can be universally quantified (instead of the admissible witnesses). For instance, the system in \cite{Altenkirch2001} forbids quantifier nestings, while the system in \cite{Leivant2001} only allows quantification $\forall X.A$ when $X$ occurs at depth at most 2 in $A$ (i.e.~when $X$ occurs at most twice to the left of an implication). Interestingly, both systems have the expressive power of G\"odel's System T (which is not a first-order system).

Another kind of restrictions on the shape of types have been investigated by the authors in \cite{Pistone2021, Pistone2020b}, motivated by ideas from the categorical semantics of polymorphism \cite{Bainbridge1990}. The two resulting fragments $\mathrm F^{\kappa\leq 0},\mathrm F^{\kappa\leq 1}$ are equivalent, respectively, to the simply typed $\lambda$-calculus with finite sums and products, and to its extension with least and greatest fixpoints (in particular, (CE) is decidable in $\mathrm F^{\kappa\leq 0}$).

Finally, polymorphism in \emph{linear} type systems has been investigated too. Interestingly, (TI) \cite{lafontMLL, lafontMALL} and (CE) \cite{LICS} remain undecidable even in this case.

 %
%
%families of subsystem  kind of restrictions of System F have been investigated. 
%
%
% example, the systems from \cite{} and \cite{} 
%
%
%- Other restrictions of polymorphism: 
%there are many!
%	- Altenchirk --> complete polymorphism,
%	- Leivant Peano's Lambda Calculus, 
%	- Finite Characteristic --> decidability of CE and TI
%	- Linear Polymorphism (still TI and CE are undecidable)

\subparagraph*{Future work}
The main interest we found in investigating $\Nda$ was to shed some new light on the source of undecidability of type-related properties for full System F. Yet, one might well ask whether the decidability of type-checking makes $\Nda$ a reasonable candidate for implementations. 
Admittedly, our decision algorithm, which was only oriented to prove decidability, is not very practical: checking failure is $\mathsf{coNP}$ with respect to the number of type symbols. Yet, it does not seems unlikely that more optimized algorithms can be developed.

By the way, given that the terms typable in $\Nda$ are simply typable, would an implementation of atomic polymorphism be interesting at all? 
In contrast with $\ML$, type-checking atomically polymorphic programs is decidable at any rank. 
One could thus investigate extensions of $\ML$ with first class atomic polymorphism (realistically, in presence of other type constructors like e.g.~some restricted version of dependent types, see \cite{Xi1999}).

A more interesting direction, suggested by our decision algorithm, would be to investigate systems with full, impredicative, polymorphism, but obeying some condition ensuring \emph{acyclicity}, so that TC (based on SOU) remains decidable. One would thus retain some advantages of first-class polymorphism (e.g.~the modularity/genericity of programs) while admitting self-applications only in ``$\ML$-style'' (or with explicit type annotations, as in $\mathrm{ML}^{\mathrm F}$ \cite{Remy2003}).
For instance, a way to ensure acycliclity might be to require that a polymorphic $\lambda$-abstracted variable be used in an \emph{affine} way, i.e.~at most once.

\bibliography{UndecidabilityEXT.bib}

\longv{
\appendix

\section{$\Nda$-unification}\label{appB}
% !TEX root = Undecidability.tex

In this section we describe a decidable unification problem, that we call $\Nda$-unification, and we show that this problem captures type-checking for $\Nda$.

\subparagraph*{A decidable second-order unification problem}

We consider a second-order language composed out of three different sorts of variables: \emph{sequence variables} $a,b,c,\dots$,
\emph{projection variables} $\alpha^{n},\beta^{n},\gamma^{n},\dots$ and \emph{second-order variables} $\SF F^{n},\SF G^{n},\dots$  (where in the last two cases $n$ indicates the arity of the variable).
The language includes expressions of three sorts, noted $\langle *\rangle$, $*$ and $T(*)$; the expressions of each type are defined by the grammars below:
\begin{align*}
\F a, \F b, \F c & ::= \langle X_{1}\dots X_{n}\rangle\mid a  \mid
%\mid \pi^{l}(a)\mid 
% \pi^{l}(
 \alpha^{n} a_{1}\dots a_{n}
   &  \qquad (\text{sort }\langle *\rangle) \\
\phi,\psi & ::=  X\mid \pi^{l}(\F a) \mid
% \alpha_{n,i}a_{1}\dots a_{n}
% \mid
  \SF F^{n}\F a_{1}\dots \F a_{n} \mid \Phi \To \Psi
  %\mid (\forall a.\Phi[a])\To \psi  
  &  \qquad (\text{sort } *) \\
\Phi,\Psi & ::= \forall a.\phi%\mid \forall a.\Phi \To \Psi
 &  \qquad (\text{sort }T( *))
\end{align*}
%For all $\Phi$ of type $T(*)$ and sequence variable $a$, we define the expression $\Phi[a]$, of type $T(*)$, by induction as follows:
%\begin{align*}
%(\forall b.X)[a] & = \forall b.X \\
%(\forall b.\pi^{l}(c))[a]&= \forall b.\pi^{l}(c) \\
%(\forall b.\pi^{l}(\alpha a_{1}\dots a_{n}))[a]&= \forall b.\pi^{l}(\alpha a_{1}\dots a_{n}a) \\
%(\forall b.\SF F\F a_{1}\dots \F a_{n})[a] & = \forall b. \SF F\F a_{1}\dots \F a_{n}a \\
%(\forall b.\phi\To \Psi)[a]& = 
%\forall b.\phi[a]\to \Psi[a]
%\\
%(\forall b.\Phi\to \Psi)[a]& = 
%\forall b. \Phi[a]\to \Psi[a]
%%
%%((\forall b. \Phi[c])\to \psi)[a]& = (\forall b.\Phi[c])[a]\to \psi[a] 
%\end{align*}

A \emph{$\Nda$-unification problem} is a pair $(U,E)$, where $U$ is a set of equations of the form $t=u$ between expressions of  type $*$, and $E$ is a set of \emph{constraints} of the form $(\alpha: a)$ or $(a:k)$, where $k\in \BB N$. 

Given a $\Nda$-unification problem $(U,E)$, for all projection variable $\alpha^{n}$ occurring in $U$, let 
$\deg(\alpha)$ indicate the maximum value $u$ such that $\pi^{u}(\alpha^{n}a_{1}\dots a_{n})$ occurs in $U$.

A \emph{substitution} for a $\Nda$-unification problem $(U,E)$ is given by the following data:
\begin{itemize}
\item for each sequence variable $a$, a natural number $k_{a}^{S}\in \BB N$;
\item for each projection variable $\alpha^{n}$, a pair $(k_{\alpha}^{S}, S(\alpha))$ made of a natural number $k_{\alpha}^{S}\geq \deg(\alpha)$ and a sequence $S(\alpha)=\langle S(\alpha)_{1},\dots, S(\alpha)_{ k_{\alpha}^{S}}\rangle$, where $S(\alpha)_{i}$ is either of the form 
$\lambda x_{1}.\dots.x_{n}.X$ or of the form $\lambda x_{1}.\dots.x_{n}.\pi^{l}(x_{j})$, where $l$ is such that, whenever
$\pi^{u}(\alpha^{n} a_{1}\dots a_{n})$ occurs in $U$, $l\leq k_{a_{j}}^{S}$;

\item for each second-order variable $\SF F^{n}$, a function $S(\SF F)$ of the form 
$\lambda \rho_{1}.\dots.\rho_{n}.  A(\rho_{1},\dots, \rho_{n})$, where 
%\begin{itemize}
%\item  $n$ is such that, if
%$\SF F\F a_{1}\dots \F a_{m}$ occurs in $U$, then $n\geq m$;
%
%% \sum_{i=1}^{m}k_{\F a_{i}}^{S} $
%
%
%\item
 $A(\rho_{1},\dots, \rho_{n})$ is given by the grammar
$$
A,B::= X\mid \pi^{l}(\rho_{i}) \mid A\To B \mid \forall X.A
$$
%
%
% from type variables to types
%
% which is either of the form 
%$\lambda x_{1}.\dots.x_{n}.X$ or of the form $\lambda x_{1}.\dots.x_{n}.\pi^{l}(x_{j})$
with $i\leq n$ and $l$ such that, if $\SF F^{n}\F a_{1}\dots \F a_{n}$ occurs in $U$, then $l \leq k_{\F a_{i}}^{S}$   (where $k_{\F a}^{S}$ is 
$k$ if $\F a=\langle X_{1},\dots, X_{k}\rangle$, is $k_{a}^{S}$ if $\F a=a$, and is $k_{\alpha}^{S}$ if $\F a=\alpha^{r} a_{1}\dots a_{r}$).
%
%
%$i=( \sum_{j=1}^{r}k_{\F a_{j}}^{S})+q$, then $

%\end{itemize}

%
%
% $l$ are such that:
%	\begin{itemize}
%	\item if
%$\SF F\F a_{1}\dots \F a_{m}$ occurs in $U$, then $n\geq \sum_{i=1}^{m}k_{\F a_{i}}^{S} $ (where $k_{\F a}^{S}$ is 
%$k$ if $\F a=\langle X_{1},\dots, X_{k}\rangle$,  $k_{a}^{S}$ if $\F a=a$, and $k_{\alpha}^{S}$ if $\F a=\alpha a_{1}\dots a_{r}$),
%	\item if 
%$\F a_{i}= a$, then $l\leq k_{a}^{S}$,
%	\item  if $\F a_{i}= \alpha a_{1}\dots a_{r}$, then 
%$l\leq k_{\alpha}^{S}$.
%	\end{itemize}
\end{itemize}

Given a substitution $S$, we define (1) for any expression $\F a$ of sort $\langle *\rangle$, a sequence $S(\F a)$ of type variables, (2) for any expression $\phi$ of sort $*$, a type $S(\phi)$, and (3) for any expression $\Phi$ of sort $T(*)$, a type $S(\Phi)$ as follows:
%we define $S(a)$,$S(\F a)_{i}$,
%$S(\phi)$ and $S(\Phi)$ as follows:
\begin{itemize}
\item if $\F a=a$,  
 $S(a)$ is a sequence of pairwise distinct variables $\langle S(a)_{1},\dots, S(a)_{k_{a}}\rangle$ (chosen in such a way that if $a\neq b$, $S(a)$ and $S(b)$ are disjoint);

\item if $\F a=\langle X_{1},\dots, X_{r}\rangle$, then $S(\F a)=\langle X_{1},\dots, X_{r}\rangle$;
\item if $\F a= \alpha^{n} a_{1}\dots a_{n}$, then 
$S(\F a)= \langle U_{1},\dots, U_{k_{\alpha}^{S}}\rangle$ where for all $i\leq k_{\alpha}^{S}$:
	\begin{itemize}
	\item if $S(\alpha)_{i}= \lambda \vec x.X$, then $U_{i}=X$;
	\item if $S(\alpha)_{i}= \lambda \vec x.\pi^{l}(x_{j})$, then $U_{i}= S(a_{j})_{l}$;

	\end{itemize}
%$S(\F a)=
%
% $i\leq n$, then 
%$S(\F a)_{i}=X$;
%\item if $\F a=\alpha a_{1}\dots a_{n}$, $\alpha_{i}^{S}= \lambda \vec x.\pi^{l}(x_{j})$ and $i\leq n$, then 
%$S(\F a)_{i}= (a_{j})^{S,i}$;

%\item if $\phi= \alpha_{n,i}a_{1}\dots a_{n}$ and $\phi^{S}= \lambda x_{1}\dots x_{n}.X$, then 
%$\phi^{S}=X$;
%\item if $\phi= \alpha_{n,i}a_{1}\dots a_{n}$ and $\phi^{S}= \lambda x_{1}\dots x_{n}.\pi^{l}(x_{j})$, then 
%$\phi^{S}=a^{S,i}_{j}$;

\item if $\phi=X$, then $S(\phi)=X$;

\item if $\phi=\pi^{l}(\F a)$, then $S(\phi)= S(\F a)_{l}$;

\item if $\phi= \SF F\F a_{1}\dots \F a_{n}$, and $S(\SF F)=\lambda \vec \rho.A$, then $S(\phi)= A[ \pi^{l}(\rho_{i})\mapsto  S(\F a_{i})_{l}]$;
%
%\item if $\phi= \SF F^{n}\F a_{1}\dots \F a_{n}$, and $\SF F^{S}=\lambda \vec x.\pi^{l}(x_{j})$, then $\phi^{S}= 
%\pi^{l}(\F a_{i}^{S})$;

\item if $\phi=\Phi\To \Psi$, then $S(\phi)=S(\Phi)\To S(\Psi)$;
\item if $\Phi=\forall a.\phi$, then $S(\Phi)=\forall S(a).S(\phi)$.

%
%
%\item if $\Phi=\forall a.\Phi_{1}\to \Phi_{2}$, then $\Phi^{S}= \forall \vec{a^{S}}.\Phi_{1}^{S}\to \Phi_{2}^{S}$.
\end{itemize}

A substitution $S$ for $(U,E)$ is a \emph{unifier} of $(U,E)$ if the following hold:
\begin{enumerate}
\item for any equation $t=u\in U$, $S(t)=S(u)$ holds;

\item for any constraint of the form $\alpha: a\in E$, $k_{a}^{S}= k_{\alpha}^{S}$;

\item for any constraint of the form $a: k\in E$, $k_{a}^{S}=k$.

\end{enumerate}

%Observe that we can suppose w.l.o.g. that, if $S$ is a unifier of $(U,E)$, then $k_{\alpha}^{S}$ is less or equal than the maximum number of occurrences of $\alpha$ in $U$.

We let \emph{$\mathtt{Fat}$-$\mathtt{unification}$} indicate the problem of finding a unifier for a $\Nda$-unification problem.
The rest of this subsection is devoted to establish the following:

\begin{theorem}
$\mathtt{Fat}$-$\mathtt{unification}$ is decidable. 

\end{theorem}

A $\Nda$-unification problem $(U,E)$ is in \emph{normal form} if if contains no equation of the form $\Phi_{1}\To \Psi_{1}=\Phi_{2}\To \Psi_{2}$.
% and no equation of the form
%$(\forall a.\Phi a)\to \psi= (\forall a'.\Phi')\to \psi'$. 
Any unification problem can be put in normal form by 
repeatedly applying the following simplification rule:
%\emph{simplifying} all equations of the form $\Phi\to \Psi$ by applying the following rules:
\begin{align*}
%& \AXC{$U+ \{X=\Phi\to \Psi\}$}
%\UIC{$U+\{X=Y\}$}
%\DP\\
%& \AXC{$U+\{\pi^{l}(\F a)=\Phi\to \Psi\}$}
%\UIC{$U+\{ X=Y\}$}
%\DP \\
%& \AXC{$U+\{\SF F\F a_{1}\dots \F a_{n}=\Phi\to \Psi\}$}
%\UIC{$U\Big [ \SF F\vec{\F a} \mapsto (\SF F_{1}\vec{\F a}\to \SF F_{2}\vec{\F a})   \Big]  +\left  \{\begin{matrix}\SF F_{1}\F a_{1}\dots \F a_{n}= \Phi, \\ \SF F_{2}\F a_{1}\dots \F a_{n}= \Psi\end{matrix}\right\}$}
%\DP\\
%&
\AXC{$U+\{(\forall a_{1}.\phi_{1})\to(\forall b_{1}. \psi_{1})= (\forall a_{2}.\phi_{2})\to(\forall b_{2}. \psi_{2})\}$}
\UIC{$\big (U+\{\phi_{1}=\phi_{2}, \psi_{1}= \psi_{2}\}\big) \big[ a_{2}\mapsto a_{1}, b_{2}\mapsto b_{1} \big]$}
\DP
\end{align*}
%where at each application of the rule one adds to $E$ all constraints of the form $(a_{1}:k)$ (resp.~$(\alpha:a_{1})$) such that $(a_{2}:k)\in E$ (resp.~$(\alpha,a_{2})\in E$), and conversely one adds to $E$ all constraints of the form $(a_{2}:k)$ (resp.~$(\alpha:a_{2})$) such that $(a_{1}:k)\in E$ (resp.~$(\alpha,a_{1})\in E$), and similarly for $b_{1},b_{2}$.   

Given a $\Nda$-unification problem in normal form $(U,E)$, we say that
an equation $t=u$ \emph{can be deduced from $U$} if $t=u$ can be deduced from a finite set of equations in $U$ by applying standard first-order equality rules.
We say that two second-order variables $\SF F,\SF G$ are \emph{equivalent} (noted $\SF F\simeq \SF G$) if an equation of the form $ \SF F\F a_{1}\dots \F a_{n}=\SF G\F b_{1}\dots \F b_{n}$ can be deduced from $U$; we say that $\SF F$ is \emph{connected with $\SF G$} (noted $\SF F \leadsto \SF G$) if
an equation of the form $\SF F\F a_{1}\dots \F a_{n}= \Phi\To \Psi$, where $\SF U$ occurs in $\Phi\To \Psi$, can be deduced from $U$.
We say that $(U,E)$ has a \emph{variable cycle} if there exist variables $\SF F_{1},\dots, \SF F_{k}$ such that 
$\SF F_{1}\stackrel{\simeq}{\leadsto}\SF F_{2}\stackrel{\simeq}{\leadsto}\dots \stackrel{\simeq}{\leadsto}\SF F_{n}\stackrel{\simeq}{\leadsto}\SF F_{1}$ (where $\SF F\stackrel{\simeq}{\leadsto} \SF G$ means that $\SF F$ is connected with some variable equivalent to $\SF G$).

%
% a directed graph $G_{U,E}$ as follows: the nodes of $G_{U,E}$ are the second-order variables $\SF F$ occurring in $U$;
% an edge $\SF F\to \SF G$ exists if $\SF F\leadsto \simeq^{*}\SF G$.
%

\begin{lemma}
Let $(U,E)$ be a unification problem in normal form. 
If $(U,E)$ has a variable cycle, then it has no solution.

\end{lemma}
\begin{proof}
To prove the lemma we show that any unification problem $(U,E)$ yields a \emph{first-order} unification problem $U^{*}$ and that any unifier of $(U,E)$ yields a unifier of $U^{*}$. For the translation, we fix a constant $c$, and  we associate any second-order variable $\SF F$ with a first-order variable $x_{\SF F}$; any expression is translated into a first order expression by: 
\begin{align*}
\F a^{*}& = c \\
%(\alpha_{n,i}a_{1}\dots a_{n})^{*}& = c \\
\SF F^{n}\F a_{1}\dots \F a_{n}& =  x_{\SF F} \\
(\Phi\To \Psi)^{*}& = \Phi^{*} \To \Psi^{*} \\
(\forall a.\phi)^{*}& = \phi^{*}
\end{align*}

We finally let
$U^{*}= \{t^{*}=u^{*}\mid t=u\in U\}$. 
Observe that if $\SF F\simeq \SF G$ in $U$, then $x_{\SF F}= x_{\SF G}$ in $U^{*}$, and if $\SF F\leadsto \SF G$ in $U$, then $U^{*}$ contains an equation of the form $x_{\SF F}= t\To u$, where $x_{\SF G}$ occurs in $t\To u$.
Hence a variable cycle in $(U,E)$ induces a variable cycle in $U^{*}$.
%Recall that a first-order unification problem with a variable cycle admits no unifier.

For any substitution $S$ for $(U,E)$, we define a first-order substitution $S^{*}$ as follows:
given $\lambda \vec \rho.A$ we define $A^{*}$ by 
$X^{*}=c$, $(\pi^{l}(\rho_{i}))^{*}=c$, $(A\To B)^{*}=A^{*}\To B^{*}$ and $(\forall X.A)^{*}=A^{*}$. 
We let then $S^{*}(x_{\SF F})=S(\SF F)^{*}$.

One can easily check that if $S$ is a unifier for $(U,E)$, then $S^{*}$ is a unifier of $U^{*}$.
As a consequence, if $(U,E)$ has a variable cycle, so does $U^{*}$, and by well-known facts about first-order unification, $U^{*}$ has no unifier, and so neither $(U,E)$ does.
\end{proof}

Let us call a unification problem $(U,E)$ \emph{simple} if it contains no expression of the form $\Phi\To \Psi$. If $(U,E)$ has no variable cycle, then it can be reduced to a simple unification problem by applying the following rules:
\begin{align*}
& \AXC{$U+ \{X=\Phi\To \Psi\}$}
\UIC{$\{X=Y\}$}
\DP\qquad
& \AXC{$U+\{\pi^{l}(\F a)=\Phi\To \Psi\}$}
\UIC{$\{ X=Y\}$}
\DP 
\end{align*}
\begin{align*}
& \AXC{$U+\{\SF F^{n}\F a_{1}^{1}\dots \F a_{n}^{1}=(\forall c_{1}.\phi_{1})\To (\forall d_{1}.\psi_{1}),\dots, \SF F^{n}\F a_{1}^{r}\dots \F a_{n}^{r}=(\forall c_{r}.\phi_{r})\To (\forall d_{r}.\psi_{r})\}$}
\UIC{$U
\Big [ \SF F^{n}\vec{\F a} \mapsto (\SF F_{1}^{n+1}\vec{\F a}c\To \SF F_{2}^{n+1}\vec{\F a}d)   \Big]  
+\left  \{\begin{matrix}\SF F_{1}^{n+1}\F a_{1}^{1}\dots \F a_{n}^{1}c_{1}= \phi_{1},\dots, \SF F_{1}^{n+1}\F a_{1}^{r}\dots \F a_{n}^{r}c_{r}= \phi_{r} \\ \SF F_{2}^{n+1}\F a_{1}^{1}\dots \F a_{n}^{1}d_{1}= \psi_{1},\dots, \SF F_{2}^{n+1}\F a_{1}^{r}\dots \F a_{n}^{r}d_{r}= \psi_{r}\end{matrix}\right\}$}
\DP
%%&
%\AXC{$U+\{\Phi_{1}\to \Psi_{1}= \Phi_{2}\to \Psi_{2}\}$}
%\UIC{$U+\{\Phi_{1}=\Phi_{2}, \Psi_{1}= \Psi_{2}\}$}
%\DP
\end{align*}
Where in the first two rules we let $Y$ be any type variable distinct from $X$, and in the last rule we suppose that $U$ contains no equation of the form
$\SF F^{n}\F a_{1}\dots \F a_{n}=\Phi\To \Psi$. Observe that, by acycliclity, $\SF F$ cannot occur in either $\phi_{i}$ or $\psi_{i}$. 
 One can argue by induction on the well-founded preorder $\stackrel{\simeq}{\leadsto}$ that one can eliminate all terms of the form $\Phi\To \Psi$ by applying a finite number of instances of the rules above.

The last step to ensure decidability is the following:
\begin{proposition}
There is an algorithm that generates all unifiers of a simple unification problem, if there exists any, and returns $\mathtt{failure}$ otherwise.
 \end{proposition}
\begin{proof}
%$(U,E)$ contains equations of type $\langle *\rangle$ and type $*$, as well as constraints $E$. 

We first describe an algorithm which generates non-deterministically all unifiers. 
We will then show that the search space for unifiers can be restricted to a finite one, allowing thus also to check if no unifier exists.

The algorithm is as follows:

\begin{enumerate}
%\item eliminate all basic equations of the form $a=b$, $a=X$ or$\pi^{l}(a)=X$, by substitution;

\item 
%Let $K$ be the maximum number of distinct occurrences of an projection variable $\alpha$ in $U$. 
Proceed non-deterministically by applying the following rules:

	\begin{enumerate}
	\item if $\pi^{j}(\alpha^{n} a_{1}\dots a_{n}) = u$ occurs in $U$, then choose a natural number $q$, set $S(\alpha)_{j}=\lambda \vec x.\pi^{q}(a_{j})$, eliminate this equation and replace in all other equations any term  of the form $ \alpha^{n} b_{1}\dots b_{n}$ with $\pi^{q}(b_{j})$;\label{stepa}
%		\item if $\pi^{j}(\alpha_{n}a_{1}\dots a_{n}) = X$, then set $\alpha_{j}^{S}=X$ and replace any term in $U$ of the form $\pi^{l}( \alpha_{n})b_{1}\dots b_{n}$ by $X$;

				\item if $\pi^{j}(\alpha^{n} a_{1}\dots a_{n}) = u$, then set $S(\alpha)_{j}=X$, where $X$ is either a variable occurring in $U$ or a fresh variable, eliminate this equation and replace in all other equations any term  of the form $\pi^{l}( \alpha^{n} b_{1}\dots b_{n})$ by $X$;

	\item if $\SF F^{n} \F a_{1}\dots \F a_{n}=u$ is in $U$, then choose a natural number $q$ and set $S(\SF F)= \lambda\vec x.\pi^{q}(x_{j})$, eliminate this equation and replace in all other equations any term  of the form $ \SF F^{n} \F b_{1}\dots \F b_{n}$ by $\pi^{q}(\F b_{j})$;\label{stepc}

%		\item if $\SF F^{n}\F a_{1}\dots \F a_{n}=X$ is in $U$, then set $\alpha_{i}^{S}= \lambda\vec x.X$ and replace any term in $U$ of the form $ \SF F^{n}\F b_{1}\dots \F b_{n}$ by $X$;
				\item if $\SF F^{n}\F a_{1}\dots \F a_{n}=u$ is in $U$, then set $S(\SF F)= \lambda\vec x.X$, where $X$ is either a variable already occurring in $U$ or a fresh variable, eliminate this equation and replace in all other equations any term  of the form $ \SF F^{n} \F b_{1}\dots \F b_{n}$ by $X$.

	\end{enumerate}

\item Observe that any rule eliminates one equation, so after a finite number of steps one is left with only equations of the form
$\F a=\F b$, where $\F a, \F b$ are either type variables or of the form $\pi^{l}(b)$, for some sequence variable $b$.
If from such equations one can deduce either $X=Y$ for distinct type variables $X,Y$, or
$\pi^{l}(a)=\pi^{r}(b)$, where either $a\neq b$ or $a=b$ and $l\neq r$, then return $\mathtt{local\ failure}$;

%  
% number of equations  application of the rule eliminates a projection or second-order variable, hence after a finite number of steps, step 1. is terminated.
%
%\item 
%
%after all projection and second-order variables are eliminated, eliminate new basic equations of form $a=b$, $a=X$ or $\pi^{l}(a)=X$ by substitution; 

\item
otherwise, set $k_{a}^{S}$ any number greater or equal to all expression $\pi^{l}(a)$ occurred during the computation, and 
set $k_{\alpha}^{S}$ is a similar way. Adjust $k_{a}^{S}, k_{\alpha}^{S}$, if possible, in order to satisfy the constraints in $E$. If this is not possible (i.e.~if $E$ contains $(a:k)$ and an expression $\pi^{l}(a)$, with $l>k$, occurred in the computation, or if $E$ further contains $(\alpha:a)$ and an expression
$\pi^{l}(\alpha^{r} a_{1}\dots a_{r})$, with $l>k$, occurred during the computation), return $\mathtt{local\ failure}$. Otherwise, for all $l\leq k_{\alpha}^{S}$ such that $\pi^{l}(a)$ never occurred, choose $S(\alpha)_{l}$ randomly;

\item if some non-deterministic branch generated by steps 2.~and 3.~ends without returning $\mathtt{local\ failure}$, then return $\mathtt{success}$.

%if all branches return $\mathtt{success}$ (observe that they are finitely many), return $\mathtt{failure}$.

\end{enumerate}

It is clear that any succeeding path yields a unifier for $(U,E)$. Conversely, if $S$ is a unifier of $(U,E)$, then one can easily see, by induction on the number of equations of $U$, that $S$ is produced by some succeeding path of the algorithm.

Now suppose that $S$ is a unifier and that $S$ is produced through a finite number of computation steps $s_{1},\dots, s_{N}$, where $N\leq \mathsf{card}(U)$. 
Let $K$ be the smallest natural number such that for all sequence $\langle X_{1},\dots, X_{r}\rangle$ in $U$ and equation of the form $\pi^{l}(\F a)=\phi$ in $U$, $r, l\leq K$, and for all constraint of the form $(a:k)$, $k\leq K$.
Then we can replace each choice of a natural number $q$ at steps \ref{stepa} and \ref{stepc} by some $q'\leq K+N$ (in fact we need no more than $N$ different values distinct from any index already occurring somewhere in $U$).

Thus we can define a new algorithm which runs all (finitely many) computations with search-space restricted to $K+N$, and returns $\mathtt{success}$ if one of them succeeds, and returns $\mathtt{failure}$ if all of them 
 return a $\mathtt{local\ failure}$.
\end{proof}

\subparagraph*{Type-checking $\Nda$ by second-order unification}

A \emph{type-checking problem} is a triple $(\Gamma, t, A)$ where $\Gamma$ is a term context, $t$ is a $\lambda$-term with $FV(t)\subseteq \Gamma$ and $A$ is a type. 
A \emph{$\Nda$-solution} of a type-checking problem is a type derivation in $\Nda$ of $\Gamma \vdash t:A$. We wish to prove the following:

\begin{theorem}
For any type-checking problem $(\Gamma, t,A)$, there exists a $\Nda$-unification problem $\B V(\Gamma, t ,A)$ such that 
$(\Gamma,t,A)$ has a solution in $\Nda$ iff $\B V(\Gamma, t,A)$ has a unifier.
\end{theorem}

The first step is to associate with each term $t$ finite sets of sequence variables, projection variables and second-order variables as follows (we suppose that no variable occurs both free and bound in $t$, and that any bound variable is bound exactly once):

%For each subterm $u$ of $t$, we define the \emph{depth of $t$ in $u$}, noted $d_{t}(u)$, as follows:
%if $t=u$, then $d_{t}(u)=0$, if $t=\lambda x.t'$, where $u$ is a subterm of $t'$, then $d_{t}(u)=d_{t'}(u)$; if $t=t_{1}t_{2}$, where $u$ is a subterm of $t_{i}$, then $d_{t}(u)=d_{t_{i}}(u)+1$. 
%Moreover, for each variable $x$, we define the \emph{depth of $x$ in $t$}, noted $d_{t}(x)$, as follows:
%if $t$ is free in $x$, then $d_{t}(x)=0$; otherwise, let $\lambda x.u$ the unique subterm of $t$ where $x$ is bound, and let $d_{t}(x)=d_{t}(\lambda x.u)$. 
%
%We now define sequence variables, projection variables and second-order variables as follows:
%
%\lambda \vec x. t_{1}\dots t_{k}$ and $u$ is a sub-term of $t_{i}$, then 
%$d_{t}(u)=d_{t_{i}}(u)+1$. 

\begin{itemize}
\item with each variable $x$ in $t$, we associate two sequence variables $a_{x}, b_{x}$, a projection variable $\alpha^{1}_{x}$, and two second-order variables $\SF F^{1}_{x}, \SF G^{1}_{x}$;

\item with each subterm of $t$ of the form $uv$, we similarly  associate two sequence variables $a_{uv}, b_{uv}$, a projection variable $\alpha_{uv}^{1}$ and two second-order variables $\SF F^{2}_{uv}, \SF G^{1}_{uv}$;

\item with each subterm of $t$ of the form $\lambda x.u$, we associate a sequence variable $b_{\lambda x.u}$, and a second order variable $\SF G^{1}_{\lambda x.t}$.

\end{itemize}

Given a set of equations $U$ and a sequence variable $a$ not occurring in $U$, we let $Ua$ be the set of equations obtained by replacing all terms $\alpha^{n} a_{1}\dots a_{n}$ by $\alpha^{n+1} a_{1}\dots a_{n}a$ and all term
$\SF F^{n} \F a_{1}\dots \F a_{n}$ by $\SF F^{n+1}\F a_{1}\dots \F a_{n} a$.

We define a set of equations $\B U(t)$, by induction on $t$
 as follows:
\begin{itemize}

\item $\B U(x)$ is formed by the equation
\begin{align*}
%\Sigma(x)& = \forall a_{x}.\SF F_{x} \\
%\SF F&= \forall b_{x}.\SF G_{x} \\
\SF F_{x} (\alpha_{x}b_{x} )& = \SF G_{x}b_{x}
\end{align*}
 %a_{y_{1}} \dots a_{y_{n}}

\item $\B U( \lambda x.t)$ is formed by $\B U(t)b_{\lambda x.t}$
%\Sigma\cup\{ x: \forall a_{x}.\SF F_{x}\}, t, \forall b_{\lambda x.t}.\SF G_{t})$
 plus the equations
\begin{align*}
\SF G_{\lambda x.t}b_{\lambda x.t}& = (\forall a_{x}.\SF F_{x}a_{x}\vec b b_{\lambda x.t}) \To \forall b_{t}.\SF G_{t}b_{t}b_{\lambda x.t}
\end{align*}

\item $\B U( tu)$ is formed by $\B U(t)b_{tu}, \B U(u)b_{tu}$ plus the equations:
\begin{align*}
\SF G_{t}b_{t}b_{tu} & = ( \forall b_{u}.\SF G_{u}b_{u}b_{tu})\To (\forall a_{tu}. \SF F_{tu}a_{tu}b_{tu})\\
\SF F_{tu} (\alpha_{tu}b_{tu}) b_{tu}& = \SF G_{tu}b_{tu}
\end{align*}

\end{itemize}
We let $\B V(\Gamma, t ,A)=(\B U(\Gamma, t,A), \B E(\Gamma, t, A))$, where  $\B U(\Gamma, t, A)$ is the union of $\B U( t)$ and all equations $\forall a_{x}.\SF F_{x}=\Gamma(x)$ and $\forall b_{t}.\SF G_{t}b_{t}=A$. 
$\B E(\Gamma,t,A)$ is formed by all constraints of the form $(\alpha_{x}: a_{x})$ and $(\alpha_{tu}: b_{t})$, as well as all constraints of the form $(a_{x}:k)$, where $\Gamma(x)= \forall X_{1}\dots X_{k}.C$, 
all constraints of the form $(b_{u}:0)$ where $t$ contains a subterm of the form $uv$, and the constraint 
$(b_{t}, h)$, where $A=\forall X_{1}\dots X_{h}.A'$.

To show that solving $\B V(\Gamma, t,A)$ is equivalent to checking if $\Gamma \vdash t:A$, 
as in \cite{IEEE},  we first define synthetic typing rules for Curry-style $\Nda$ as shown in Fig.~\ref{fig:synthetic}, 
where $A\preceq B$ holds when $A=\forall X_{1}\dots X_{n}.A$ and $B=A[X_{1}\mapsto Y_{1},\dots, X_{n}\mapsto Y_{n}]$.

\begin{figure}
\begin{center}
\fbox{
\begin{minipage}{\textwidth}
$$
\begin{matrix} \ \\ 
\AXC{$\Gamma(x)=A$}
\AXC{$A\preceq B$}
\RL{$\vec X \notin FV(\Gamma)$}
\BIC{$\Gamma \vdash x: \forall \vec X:B$}
\DP 
\\
\ \\
\AXC{$\Gamma, x:A\vdash t: B$}
\RL{$\vec X \notin FV(\Gamma)$}
\UIC{$\Gamma \vdash \lambda x.t: \forall \vec X.A\To B$}
\DP
\\
\ \\
\AXC{$\Gamma \vdash t:A\To B$}
\AXC{$\Gamma \vdash u:A$}
\AXC{$B\preceq C$}
\RL{$\vec X \notin FV(\Gamma)$}
\TIC{$\Gamma \vdash tu: \forall \vec X.C$}
\DP\\ \ \\
\end{matrix}
$$
\end{minipage}
}
\end{center}
\caption{Synthetic typing rules for Curry-style $\Nda$.}
\label{fig:synthetic}
\end{figure}

One can check by induction on $t$ that a synthetic type derivation of $\Gamma \vdash t:A$ yields a unifier of $\B V(\Gamma, t,A)$.
Conversely, we show that from a unifier $S$ for $\B V(\Gamma, t,A)$ we can construct a synthetic typing derivation of $\Gamma\vdash t:A$. We argue by induction on $t$:
\begin{itemize}

\item if $t=x$, then we have $\Gamma(x)=\forall X_{1}\dots X_{N}.S(\SF F_{x})\vec X$, where $N=k_{a_{x}}^{S}$, 
$A= \forall Y_{1}\dots Y_{P}.S(\SF G_{x})\vec Y$, where $P=k_{b_{x}}^{S}$, and moreover, 
$S(\SF F_{x})  (S(\alpha_{x})_{1}\vec Y)\dots (S(\alpha_{x})_{N}\vec Y)= S(\SF G_{x})\vec Y$
(using the fact that $k_{\alpha_{x}}^{S}= k_{a_{x}}^{S}=N$). 
Observe that $(S(\alpha_{x})_{j}\vec Y)$ is a variable, and we deduce then that 
$\Gamma(x) \preceq S( \SF G_{x})\vec Y$; since we can suppose that $\vec Y$ does not occur in $\Gamma$, we deduce then that 
$$
\AXC{$\Gamma(x)=\forall \vec X.S(\SF F_{x})\vec X$}
\AXC{$\forall \vec X.S(\SF F_{x})\vec X\preceq S( \SF G_{x})\vec Y$}
\RL{$\vec Y\notin FV(\Gamma)$}
\BIC{$\Gamma \vdash x: A$}
\DP
$$

\item if $t=\lambda x.u$, then we have that $A=\forall X_{1}\dots X_{N}.(\forall Y_{1}\dots Y_{P}.S(\SF F_{x})\vec Y\vec X)\To \forall Z_{1}\dots Z_{Q}.S(\SF G_{u})\vec Z\vec X$, where $N=k_{b_{\lambda x.t}}^{S}$, $P=k_{a_{x}}^{S}$, $Q=k_{b_{u}}^{S}$ and where we can suppose that the $X_{i}$ do not occur free in $\Gamma$; since  since $\SF U(\lambda x.t)=\SF U(t)a_{\lambda x.t}$ we deduce that  $S$ unifies 
$\B V( \Gamma\cup\{x:A_{1}\}, u, A_{2}))$. By I.H. we deduce then the existence of a type derivation of $\Gamma , x:A_{1}\vdash u: A_{2}$, and since the $X_{i}$ do not occur in $\Gamma$ we finally have
$$
\AXC{[I.H.]}
\noLine
\UIC{$\Gamma, x:A_{1}\vdash u:A_{2}$}
\RL{$\vec X\notin FV(\Gamma)$}
\UIC{$\Gamma\vdash t:A$}
\DP
$$

\item if $t=uv$, then we have that $A=\forall X_{1}\dots X_{n}.S(\SF G_{uv})\vec X$, 
 $S(\SF G_{u})\vec X=(\forall Y_{1}\dots Y_{P}.S(\SF G_{v})\vec Y\vec X)\To (\forall Z_{1}\dots Z_{Q}.S(\SF F_{uv})\vec Z\vec X)$ and that 
$S(\SF F_{uv})(S(\alpha_{uv})_{1}\vec X)\dots (S(\alpha_{uv})_{N}\vec X)\vec X= S(\SF G_{uv})\vec X$, where 
$N=k_{b_{uv}}^{S}$, $P=k_{b_{u}}^{S}$ and $Q=k_{a_{uv}}^{S}$, and where we use the fact that 
$k_{b_{u}}^{S}=0$. Moreover, for any choice of the variables $\vec X$, we have that 
$S$ unifies $\B V( \Gamma, u, (\forall Y_{1}\dots Y_{P}.S(\SF G_{v})\vec Y\vec X)\to\forall Z_{1}\dots Z_{Q}. S(\SF F_{uv})\vec Z\vec X)$ and
$\B V (\Gamma, v , \forall Y_{1}\dots Y_{P}.S(\SF G_{u})\vec Y\vec X)$; by choosing the $\vec X$ so that they do not occur free in $\Gamma$, using the I.H. and the fact that $k_{\alpha_{uv}}^{S}= k_{a_{uv}}^{S}=Q$, we deduce then

\medskip
\resizebox{\textwidth}{!}{$
\AXC{[I.H.]}
\noLine
\UIC{$\Gamma\vdash u:  (\forall Y_{1}\dots Y_{P}.\SF G_{v}^{S}\vec Y\vec X)\to( \forall Z_{1}\dots Z_{Q}. S(\SF F_{uv})\vec Z\vec X)$}
\AXC{[I.H.]}
\noLine
\UIC{$\Gamma\vdash v:  \forall Y_{1}\dots Y_{P}.S(\SF G_{v})\vec Y\vec X$}
\AXC{$ \forall Z_{1}\dots Z_{Q}.S(\SF F_{tu}) \preceq S(\SF G_{tu})\vec X$}
\RL{$\vec X\notin FV(\Gamma)$}
\TIC{$\Gamma \vdash t: A$}
\DP
$}

\end{itemize}

%
%\section{System $\Nda$ with constants and atomic axioms}\label{appC}
%\input{axioms}
%\input{constants}

\section{Bounded Multiplication in $\ML$}\label{appD}

We show that the numerical functions representable in $\ML$ are closed under bounded multiplication: if the function $g:\BB N\to \BB N$ is representable, then also the function $f: \BB N\to \BB N$ defined by $f(n)=\prod_{k\leq n}g(k)$ is representable.

The first step is to define an encoding of pairs inside $\STLC$. Let $A,B$ be two simple types, and let $\langle A,B\rangle:= (A\To B\To o)\To o$. We let $\langle t,u\rangle$ be the term $\lambda f.ftu$. We define the projections 
$\SF P_{1}^{A}:\langle A,B\rangle \To A$ and $\SF P_{2}^{B}:\langle A,B\rangle \To B$  by essentially following the predicative encoding of products in $\Nda$. We first define a term $\SF T^{C}: \langle A,B\rangle \To (A\To B\To C)\To C$ by induction on a simple type $C$ as follows:
\begin{itemize}
\item if $C=o$, then $\SF T^{C}= \lambda gh.gh $;
\item if $C=C_{1}\To C_{2}$, then $\SF T^{A}= \lambda ghz. \SF T^{A_{2}}g (\lambda ab.habz)$.
\end{itemize}
We finally let $\SF P_{1}^{A}:= \SF T^{A}\lambda xy.x$ and $\SF P_{2}^{B}= \SF T^{B}\lambda xy.y$.
This definition can be straightforwardly extended to $n$-ary products $\langle A_{1},\dots, A_{n}\rangle$, with projection functions $\SF P_{i,n}^{A_{i}}: \langle A_{1},\dots, A_{n}\rangle\To A_{i}$.

We furthermore exploit the following standard result (see \cite{Leivant1991}):
\begin{lemma}\label{lemma:ml}
For all $\beta$-normal term $t$, $\vdash t: \Nat \To \Nat$ holds in $\ML$ iff there exists $k\in \BB N$, a term $t^{*}$ and simple types $A_{1},\dots, A_{k}$ such that:
\begin{itemize}
\item $\vdash t^{*}:\Nat[A_{1}]\To \dots \To \Nat[A_{k}]\To \Nat[o]$ holds in $\STLC$;
\item $t^{*}x\dots x\simeq_{\beta} tx$,

\end{itemize}
where $\Nat[A]=(A\To A)\To (A\To A)$.
\end{lemma}

Let $g: \BB N \to \BB N$ be a function represented in $\ML$ by some term $t_{g}: \Nat \To \Nat$. By Lemma \ref{lemma:ml} we deduce that there exist simple types $A_{1},\dots, A_{n}$ and a term $t^{*}_{g}$ such that $t^{*}_{g}: \Nat[A_{1}]\To \dots \To \Nat[A_{n}]\To \Nat[o]$ in $\STLC$, and $t^{*}_{g} x\dots x\simeq_{\beta\eta}t_{g}x$. 
Let now $f(n)= \prod_{k\leq n}g(k)$. 

Let $C= \langle \Nat[A_{1}],\dots, \Nat[A_{n}]\rangle$ and 
$D= \langle \Nat[o], C, \Nat[o]\rangle$. 
We define the term $t_{f}$ by: 
$$t_{f}= \lambda x. \SF P_{1,3}^{\Nat[o]}\Big ( 
x 
\big (  
\lambda h. \langle v_{1}[h], v_{2}[h], v_{3}[h]\rangle\big)
%
%
% \SF P_{3,3}h)\times (\SF P_{1,3}h),  
% \langle \SF P_{2,3}(\SF P_{1,n}(h))+1,\dots,  \SF P_{2,3}(\SF P_{n,n}(h))+1                   \rangle 
%\big), 
%t^{*}_{g} (\SF P_{2,3}(\SF P_{1,n}(h)))\dots (\SF P_{2,3}(\SF P_{n,n}(h)))        
%\Big)
\langle t^{*}_{g} \B 0\dots \B 0, \langle \B 0,\dots, \B 0\rangle, t_{g}^{*}\B 0\dots \B 0\rangle
\Big)
\ : \ \Nat[D]\To \Nat[o]
$$
where 
\begin{align*}
v_{1}[h]& = ( \SF P^{\Nat[o]}_{3,3}h)\cdot (\SF P_{1,3}^{\Nat[o]}h)\\
v_{2}[h]& = \left\langle \SF P_{1,n}^{\Nat[A_{1}]}(\SF P_{2,3}^{C}(h))+1,\dots,  \SF P_{n,n}^{\Nat[A_{n}]}(\SF P_{2,3}^{C}(h))+1  \right\rangle 
\\
v_{3}[h]& =t^{*}_{g} \Big (\SF P_{1,n}^{\Nat[A_{1}]}(\SF P_{2,3}^{C}(h))+1\Big )\dots\Big (\SF P_{n,n}^{\Nat[A_{n}]}(\SF P_{2,3}^{C}(h))+1\Big )        
\end{align*}
Observe that $t_{f}\B n$ computes the triple $\langle \prod_{k\leq n} t_{g}\B k, \langle \B n, \dots, \B n \rangle, t_{g}\B n\rangle$ and then extracts its first element.
Since $t_{f}$ has type $\Nat[D]\To \Nat[o]$ in $\STLC$, we deduce, again by Lemma \ref{lemma:ml}, that $t_{f}$ has type $\Nat \To \Nat $ in $\ML$.

\section{Proofs from Section \ref{sec:equivalence}}\label{appE}
% !TEX root = Undecidability.tex

\begin{proof}[Proof of Lemma \ref{lemma2}]
%We will prove by simultaneous induction the following claims:

\begin{enumerate}
\item We must consider two cases:
\begin{enumerate}
\item if $\rr{\TT C}[\ ]= \bb{x_{i}}$, then $\rr{\TT C}[\bb{u_{A}}]= \rr{\TT C}[\bb{v_{A}}]=\bb{x_{i}}$.
\item if $\rr{\TT C}[\ ]= \rr{\TT E}[\ ]\bb{\TT C_{1}[\ ]\TT C_{2}[\ ]}$, for some $\rr{\TT E}[\ ]\in \BB G_3$ and $\TT C_{i}[\ ]\in\BB G_{1}$, then by the induction hypothesis
$\rr{\TT E}[\bb{u_{A}}]\simeq_{\beta\eta} \rr{\TT E}[\bb{v_{A}}]$ and
$\rr{\TT C_{i}}[\bb{u_{A}}]\simeq_{\beta\eta} \rr{\TT E}[\bb{v_{A}}]$, so we can conclude.

%
%$$
%\rr{\TT C}[u_{A}]=
%\Lambda X.\lambda x_1x_{2}. u_{A} \ \Lambda X.\lambda y_1y_{2}.y_{2}\star \ X x_1x_{2}
%\simeq_{\beta}\Lambda X.\lambda y_1y_{2}.y_{2}\star
%$$
%$$
%\rr{\TT C}[v_{A}]=
%\Lambda X.\lambda x_1x_{2}. v_{A} \ \Lambda X.\lambda y_1y_{2}.y_{2}\star \ X x_1x_{2}
%\simeq_{\beta}\Lambda X.\lambda y_1y_{2}.(\lambda d.y_{2}\star)\star
%\simeq_{\beta}
%\Lambda X.\lambda y_1y_{2}.y_{2}\star
%$$
\end{enumerate}

\item We must consider two cases:
\begin{enumerate}
\item if $\rr{\TT D}[\ ]= z\Lambda W_{i}.\lambda w_{i}.\rr{\TT F}[\ ]$ for some $\rr{\TT F}[\ ]\in \TT G_{4}$, then 
by the induction hypothesis
$\rr{\TT F}[\bb{u_{A}}]\simeq_{\beta\eta} \rr{\TT F}[\bb{v_{A}}]\simeq_{\beta\eta}w_{i}$, whence
$\TT D[u_{A}]\simeq_{\beta} z\SF{id}\simeq_{\beta}\TT D[v_{A}]$.

\item If $\rr{\TT D}[\ ]= \rr{\TT E}[\ ] \bb W_{i}\TT D_{1}[\ ]\TT D_{2}[\ ]$ where $\rr{\TT E}[\ ]\in \BB G_3$ and $\TT D_{i}[\ ] \in \TT G_{2}$, then by the induction hypothesis
$\rr{\TT E}[\bb{u_{A}}]\simeq_{\beta\eta} \rr{\TT E}[\bb{v_{A}}]\simeq_{\beta}\B b$, where $\B b\in \{\B t, \B f\}$,  since a closed $\beta$-normal term of type $\SF{Bool}$ must be either $\B t$ or $\B f$, 
 and 
$\rr{\TT D_{i}}[\bb{u_{A}}]\simeq_{\beta\eta} \rr{\TT D_{i}}[\bb{v_{A}}]\simeq_{\beta}z_{i}\SF{Id}$.
By letting $j=1$ if $\B b=\B t$ and $j=2$ if $\B b=\B f$, we  deduce then
$\TT D[u_{A}]\simeq_{\beta} \B b \TT D_{1}[u_{A}]\TT D_{2}[u_{A}]\simeq_{\beta}
\TT D_{j}[u_{A}]\simeq_{\beta}z_{i}\SF{Id}$ and we similarly deduce $\TT D[v_{A}]\simeq_{\beta} z_{i}\SF{Id}$.
\end{enumerate}

\item if $\TT E[\ ]\in \{\B t, \B f\}$, the claim is immediate; otherwise, if 
$\rr{\TT E}[\ ]=x\Lambda Y.\lambda y.\lambda z.\rr{\TT D}[\ ]$, where $\TT D\in \BB G_{2}$, 
% and $\TT C_{i}\in \BB G_{1}$, 
then by the induction hypothesis $\Lambda Y.\lambda y.\lambda z.\rr{\TT D}[\bb{u_{A}}]\simeq_{\beta\eta}\Lambda Y.\lambda y.\lambda z. \rr{\TT D}[\bb{v_{A}}]\simeq_{\beta\eta}\Lambda Y.\lambda y.\lambda z.z\SF{Id} =
\bb{\iota_{2}(\mathsf{Id})}$
%and $\TT C_{i}[u_{A}]\simeq_{\beta}\TT C_{i}[v_{A}]$, 
so we can compute
%\begin{equation*}
%\begin{split}
$\rr{\TT E}[\bb{u_{A}}] \simeq_{\beta}
u_{A}\iota_{2}(\SF{Id}) \simeq_{\beta} \B f$
%
%
%\Lambda Z.\lambda x x'. (u_{A} \iota_{2}(\SF{Id})\TT C_{1}[u_{A}]\TT C_{2}[u_{A}]\\
%& \simeq_{\beta}
%\Lambda Z.\lambda x x'. \B f\TT C_{1}[u_{A}]\TT C_{2}[u_{A}] \simeq_{\beta}
%\Lambda Z.\lambda x x'. \TT C_{2}[u_{A}]
%\end{split}
%\end{equation*}
% \bb{u_{A}}\rr{\TT D}[\bb{u_{A}}]\simeq_{\beta} \bb{\B f}
and
%\begin{equation*}
%\begin{split}
$\rr{\TT E}[\bb{v_{A}}] \simeq_{\beta}
\IOor_{\Bool}[\iota_{2}(\SF{Id})] (\lambda x. \B t)(\lambda x.\B f)
\simeq_{\beta} \B f$.
%
%\Lambda Z.\lambda x x'. (u_{A} \iota_{2}(\SF{Id})\TT C_{1}[u_{A}]\TT C_{2}[u_{A}] \\
%& \simeq_{\beta}
%\Lambda Z.\lambda x x'. \SF{Case}_{\SF{Bool}}[\iota_{2}(\SF{Id}] 
%\TT C_{1}[u_{A}]\TT C_{2}[u_{A}] \\
% &\simeq_{\beta}\Lambda Z.\lambda x x'. \B f\TT C_{1}[u_{A}]\TT C_{2}[u_{A}] 
%\simeq_{\beta}
%\Lambda Z.\lambda x x'. \TT C_{2}[u_{A}]
%\end{split}
%\end{equation*}

%
%\begin{equation*}
%\begin{split}
%\rr{\TT E}[\bb{v_{A}}] & \simeq_{\beta} \bb{v_{A}}\rr{\TT D}[\bb{v_{A}}]\simeq_{\beta} 
%\bb{ \mathsf{Case}_{\Bool}}[ \rr{\TT D}[\bb{v_{A}}]] \ \bb{\lambda x.\B t} \ \bb{\lambda x.\B f} \\
%& \simeq_{\beta\eta}
%\bb{ \mathsf{Case}_{\Bool}}[ \bb{\iota_{2}[\mathsf{Id}]}] \ \bb{\lambda x.\B t} \ \bb{\lambda x.\B f} 
%\simeq_{\beta\eta}^{[\text{Prop. }\ref{prop:case}]}
%\bb{(\lambda x.\B f)\mathsf{Id}} 
%\simeq_{\beta}\bb{ \B f}
%\end{split}
%\end{equation*}

\item  We must consider two cases:
\begin{enumerate}
\item If $\rr{\TT F}[\ ]=w$ then the claim is trivially true;
\item If $\rr{\TT F}[\ ]= \TT E[\ ]W \TT F_{1}[\ ] \TT F_{2}[\ ]$, with $\TT E[\ ]\in \BB G_{3}$ and $\TT F_{i}[\ ]\in \BB G_{4}$, then by the induction hypothesis 
$\TT E[u_{A}]\simeq_{\beta} \TT E[v_{A}]\simeq_{\beta}\B b$, where $\B b\in \{\B t, \B f\}$, and 
$\TT F_{i}[u_{A}]\simeq_{\beta}\TT F_{i}[v_{A}]= w_{i}$. Then we deduce 
$\TT F[u_{A}]\simeq_{\beta} \B b \TT F_{1}[u_{A}]\TT F_{2}[u_{A}]\simeq_{\beta} \B b w_{i}w_{i}\simeq_{\beta} w_{i}$ and we similarly deduce $\TT F[v_{A}]\simeq_{\beta}w_{i}$.
%
%
%\bb{\Lambda Z.\lambda z.}\rr{\TT E}[\ ]\bb{zz}$ for some $\rr{ \TT E}[\ ]\in \BB G_{3}$, then by the induction hypothesis $\rr{\TT E}[\bb{u_{A}}]\simeq_{\beta\eta} \rr{\TT E}[\bb{v_{A}}]\simeq_{\beta\eta}\bb{\B f}$, whence
%we can conclude by applying some $\beta$-reduction steps. 
\end{enumerate}
\end{enumerate}
\end{proof}

We now adapt the argument the other contextual equivalence relation $\simeq_{\Nat}^{\Nda}$.

We define $A^{*}$ and the terms $u_{A},v_{A}$ as in the previous case.
%, and we define two terms $\bb{u_{A}}, \bb{v_{A}}: (A^{*}\widetilde + \widetilde\top)\To \Nat$ as follows:
%\begin{equation*}
%\begin{split}
%\bb{u_{A}}  &  = \bb{ \lambda x. \B{0} } \\
%\bb{v_{A} } &  = \bb{ \lambda x.  \IOor_{\Nat}[\ ]   (\lambda x. \B 0)( \lambda x. \B 1)}
%\end{split}
%\end{equation*}
We will use the families of contexts $\BB G_{1}$-$\BB G_{4}$ from the other case and a new family $\BB G_{5}$ defined as follows:

\begin{minipage}{0.63\textwidth}
\begin{equation*}
\begin{split}
\BB G_{5}: \qquad \rr{\TT B}[\ ] & ::=   g \mid f \TT B[\ ]   \mid\rr{\TT E}[\ ]U \bb{\TT B[\ ]\TT B[\ ]} %\\
%\BB G_{1}: \qquad \rr{\TT C}[\ ] & =   \bb{x_{i}} \mid f_{i}\TT C[\ ]   \mid\rr{\TT E}[\ ] Y_{i}\bb{\TT C[\ ]\TT C[\ ]} \\
%\BB G_{2}:\qquad \rr{\TT D}[\ ] & =  \bb{ z}\Lambda W.\lambda w.\rr{\TT F}[\ ]\mid \bb{\TT E}[\ ]\bb{W} \TT D[\ ] \TT D[\ ]\\
%\BB G_{3}: \qquad \rr{\TT E}[\ ]&  =   \B t\mid \B f\mid
%% \Lambda Z.\lambda x.\lambda x'. ([\ ] \Lambda Y.\lambda y.\lambda z.\rr{\TT D}[\ ])\TT C_{1}[\ ]\TT C_{2}[\ ]
%x(\Lambda Y.\lambda y.\lambda z.\rr{\TT D}[\ ])
%  \\
%\BB G_{4}: \qquad \rr{\TT F}[\ ]&= w\mid \TT E[\ ]W\TT F[\ ]\TT F[\ ]
%\bb{ \mathsf{Id}} \mid\bb{ \Lambda Z.\lambda z.\TT E}[ \  ]\bb{Zzz}
\end{split}
\end{equation*}
%\end{minipage}
%\begin{minipage}{0.33\textwidth}
%\begin{equation*}
%\begin{split}
%&  (\rr{\TT E}[\ ]\in \BB G_{3})%\\
%%&  (\rr{\TT E}[\ ]\in \BB G_{3})\\
%% & (\rr{\TT E}[\ ]\in \BB G_{3}, \rr{\TT F}\in \BB G_{4}) \\
%%& (\rr{\TT D}[\ ]\in \BB G_{2}
%%%, \TT C_{i}[\ ]\in \BB G_{1}
%%)\\
%%& (\rr{\TT E}[\ ]\in \BB G_{3})
%\end{split}
%\end{equation*}
\end{minipage}

%\begin{figure}%, depending on two free variables $x_1,x_{2}$:
%\caption{Sets of contexts $\BB G_{1}-\BB G_{4}$.}
%\label{fig:contexts}
%\end{figure}

\begin{lemma}\label{lemma2bis}
For all $\rr{\TT B}[ \ ]\in \BB G_5$, $\rr{\TT B}[\bb{u_{A}}]\simeq_{\beta\eta}\rr{\TT C}[\bb{v_{A}}]$.
\end{lemma}
\begin{proof}
%We prove by simultaneous induction the following claims:
%\begin{enumerate}
%\item for all $\rr{\TT C}[\ ]\in \BB G_1$, $\rr{\TT C}[\bb{u_{A}}]\simeq_{\beta\eta}\rr{\TT C}[\bb{v_{A}}]$.
%\item for all $\rr{\TT D}[\ ]\in \BB G_2$, $\rr{\TT D}[\bb{u_{A}}]\simeq_{\beta\eta}\rr{\TT D}[\bb{v_{A}}]\simeq_{\beta\eta} \bb{z_{i} \mathsf{Id}}$.
%\item for all $\rr{\TT E}[\ ]\in \BB G_3$, $\rr{\TT E}[\bb{u_{A}}]\simeq_{\beta\eta}\rr{\TT E}[\bb{v_{A}}]$.
%%\bb{\B f}$.
%\item for all $\rr{\TT F}[\ ]\in \BB G_4$, $\rr{\TT F}[\bb{u_{A}}]\simeq_{\beta\eta}\rr{\TT F}[\bb{v_{A}}]\simeq_{\beta\eta} \bb{w_{i}}$.
%\item for all $\rr{\TT B}[\ ]\in \BB G_5$, $\rr{\TT B}[\bb{u_{A}}]\simeq_{\beta\eta}\rr{\TT B}[\bb{v_{A}}]$.
%
%
%
%
%\end{enumerate}
%The only new case is the last one, while the other can be treated as in the proof of Lemma \ref{lemma2}.
There are three cases:
\begin{enumerate}[a.]
\item if $\TT B[\ ]=g_{i}$, then $\TT B[u_{A}]=\TT B[v_{A}]=g_{i}$;
\item if $\TT B[\ ]=f\TT B'[\ ]$, where $\TT B'[\ ]\in \BB G_{5}$, then by I.H. $\TT B'[u_{a}]\simeq_{\beta\eta}\TT B'[v_{a}]$, so $\TT B[u_{A}]=f\TT B'[u_{A}]\simeq_{\beta\eta}f\TT B'[v_{A}]=\TT B[v_{A}]$;
\item if $\TT B[\ ]= \TT E[\ ]U  \TT B_{1}[\ ]\TT B_{2}[\ ]$, where $\TT E[\ ]\in \TT G_{3}$ and $\TT B_{i}[\ ]\in \TT G_{5}$, then we can similarly conclude using Lemma \ref{lemma2} and the I.H.
\end{enumerate}
\end{proof}

\begin{lemma}\label{lemma1bis}
Let $\rr{\TT K}[\ ]: (A^{*}\widetilde +\widetilde \top)\To \Bool \vdash^{f:U\To U,g\mapsto U} U$ be a $\beta$-normal term context. Then either $ A$ is provable in $\Ndac$ or $\rr{\TT K}[\ ]\in\BB G_5$.
\end{lemma}
\begin{proof}
We will prove the following claim: either there exists contexts $\Gamma,\Theta,\Delta, \Sigma$, where
\begin{equation}\label{eq:delta}
\begin{split}
\Phi& =  \{f\mapsto U\To U, g:U\}\\
\Gamma& = \{ x_{1}\mapsto Z_{1},x'_{1}\mapsto Z_{1},\dots, x_{p}\mapsto Z_{p},x'_{p}\mapsto Z_{p}\}  \\
\Theta & = \{ w_{1}\mapsto W_{1},\dots, w_{q}\mapsto W_{q}\}\\
 \Delta& =\{ y_{1}\mapsto A^{*}\To Y_{1},\dots, y_{r}\mapsto A^{*}\To Y_{r}\} \\
 \Sigma & =\{ z_{1}\mapsto \widetilde \top \To Y_{1},\dots, z_{r}\mapsto  \widetilde \top\To Y_{r}\}
\end{split}
\end{equation}
for some $p,q,r\in \BB N$ and variables $U,Z_{1},\dots, Z_{p},W_{1},\dots, W_{q},Y_{1},\dots, Y_{r}$  pairwise distinct and disjoint from $A$, 
and a context $\rr{\TT H}[\ ]: (A^{*}\widetilde +\widetilde \top)\To \Bool \vdash^{\Phi,\Gamma,\Theta,\Delta, \Sigma} A^{*}$, or $\rr{\TT K}[\ ]\in \BB G_1$.
If the main claim is true we can deduce the statement of the lemma as follows: suppose $\rr{\TT K}[\ ]\notin \BB G_1$.
Then let $\theta$ be the substitution sending all variables in $\Phi,\Gamma,\Theta,\Delta,\Sigma$ plus $Y$ onto $\clubsuit$ and being the identity on all other variables.
Then $\rr{\TT H\theta}[\ ]: ( (\clubsuit \To A)\widetilde + \widetilde \top)\To \Bool \vdash^{\Delta\theta, f\mapsto \clubsuit\To \clubsuit, g\mapsto \clubsuit}\mapsto  \clubsuit \To A$. Then we have $\Phi\theta,\Gamma\theta,\Theta\theta,\Delta\theta,\Sigma\theta\vdash \bb t:A$, where 
$
\bb t= \bb{\TT H\theta}[ \lambda x. \B t ]\bb{\star}
$ and we can conclude that $\vdash \bb{t'}:A$ holds where $\bb{t'}$ is obtained from $\bb t$ by 
substituting the variables in $\Phi$ by $\lambda x.x$ and $\star$, those in $\Gamma,\Theta$ by $\star$ and those in $\Delta$ and $\Sigma$ by $\lambda x.\star$.

Let us then prove the main claim.
Suppose that for no $\Phi,\Gamma,\Theta,\Delta, \Sigma$ there exists a context $\rr{\TT H}[\ ]: (A^{*}\widetilde +\widetilde \top)\To \Bool \vdash^{\Phi,\Gamma,\Theta,\Delta,\Sigma} A^{*}$. We will show by simultaneous induction the following claims:
\begin{enumerate}

\item for all $\Phi,\Gamma,\Theta,\Delta, \Sigma$ as in Eq.~\eqref{eq:delta}, if $\rr{\TT K}[\ ]: (A^{*}\widetilde +\widetilde \top)\To \Bool \vdash^{\Phi,\Gamma,\Theta,\Delta, \Sigma} Z_{i}$, then $\rr{\TT K}[\ ]\in \BB G_1$;
\item for all $\Phi,\Gamma,\Theta,\Delta, \Sigma$ as in Eq.~\eqref{eq:delta}, if $\rr{\TT K}[\ ]: (A^{*}\widetilde +\widetilde \top)\To \Bool \vdash^{\Phi,\Gamma,\Theta,\Delta, \Sigma} Y_{i}$, then $\rr{\TT K}[\ ]\in \BB G_2$;

\item for all $\Phi,\Gamma,\Theta,\Delta, \Sigma$ as in Eq.~\eqref{eq:delta}, if $\rr{\TT K}[\ ]: (A^{*}\widetilde +\widetilde \top)\To \Bool \vdash^{\Phi,\Gamma,\Theta,\Delta, \Sigma} \Bool$ and $\rr{\TT K}[\ ]$ is an elimination context, then $\rr{\TT E}[\ ]\in \BB G_3$;

\item for all $\Phi,\Gamma,\Theta,\Delta, \Sigma$ as in Eq.~\eqref{eq:delta}, if $\rr{\TT K}[\ ]: (A^{*}\widetilde +\widetilde \top)\To \Bool \vdash^{\Phi,\Gamma,\Theta,\Delta, \Sigma} W_{i}$, then $\rr{\TT K}[\ ]\in \BB G_4$;
\item for all $\Phi,\Gamma,\Theta,\Delta, \Sigma$ as in Eq.~\eqref{eq:delta}, if $\rr{\TT K}[\ ]: (A^{*}\widetilde +\widetilde \top)\To \Bool \vdash^{\Phi,\Gamma,\Theta,\Delta, \Sigma} U$, then $\rr{\TT K}[\ ]\in \BB G_5$.

\end{enumerate}
The main claim then follows from 1. by taking $\Gamma=\{x\mapsto Z,x'\mapsto Z\} $ and $\Theta=\Delta=\Sigma=\emptyset$.

We argue for each case separately. The last case is new, while the other four can be treated as in the proof of Lemma \ref{lemma1}. There exist three possibilities for $\TT K[\ ]$:

\begin{enumerate}[a.]
\item $\rr{\TT K}[\ ]=\bb{g}$, hence $\rr{\TT K}[\ ]\in \BB G_5$;
\item $\rr{\TT K}[\ ]=f\TT K'[\ ]$, then by I.H. $\TT K'[\ ]\in \BB G_{5}$, so $\TT K[\ ]\in \BB G_{5}$;
\item $\rr{\TT K}[\ ]=  \rr{\TT K'}[\ ] \bb{U  \TT K_{1}[\ ]\TT K_{2}[\ ]}$, where $\rr{\TT K'}[\ ]: (A^{*}\widetilde +\widetilde \top)\To \Bool \vdash^{\Phi,\Gamma,\Theta,\Delta,\Sigma} \Bool$,  
 $\rr{\TT K_{i}}[\ ]: (A^{*}\widetilde +\widetilde \top)\To \Bool \vdash^{\Phi,\Gamma,\Theta,\Delta,\Sigma} U$, 
  %$\rr{\TT B_{2}}[\ ]: (A^{*}\widetilde +\widetilde \top)\To \Bool \vdash^{\Phi,\Gamma,\Theta,\Delta,\Sigma} U$, 
  and where $\TT K'[\ ]$ is an elimination context. By the induction hypothesis then $\rr{\TT K'}[\ ]\in \BB G_3, \TT K_{i}[\ ]\in \BB G_{5}$, hence $\rr{\TT K}[\ ]\in \BB G_5$.

\end{enumerate}
\end{proof}

\begin{proposition}\label{prop:sbobis}
$\bb{u_{A}}\not\simeq_{\Nat}^{\Nda}\bb{v_{A}}$ iff $A$ is provable in $\Ndac$.
\end{proposition}
\begin{proof}
The argument is similar to the proof of Proposition \ref{prop:sbo}.
%Suppose $\vdash t:A$ is derivable in $\Ndac$, then by letting $\bb{t^{*}}= \bb{\lambda y. t}[\bb y/\bb{\star}]$ we deduce
%$\vdash\bb{t^{*}}: A^{*}$ and we can show $\bb{u_{A}}\not\simeq_{\Bool}^{\Nda}\bb{v_{A}}$, by letting
%$
%\rr{\TT E}[\ ] =  [\ ] \bb{\iota_{1}[t^{*}]}
%$, 
%since $\rr{\TT E}[ \bb{u_{A}}]\simeq_{\beta}\bb{\B f} $ and
%$\rr{\TT E}[\bb{v_{A}}]\simeq_{\beta}
%\bb{\mathsf{Case}_{\Bool}[\iota_{1}[t^{*}]] \ \lambda x.\B t \ \lambda x.\B f} \simeq_{\beta\eta}
%\bb{(\lambda x.\B t) t^{*}} \simeq_{\beta} \bb{\B t}$.
%
%Conversely, suppose $A$ is not provable in $\Ndac$.
%Any context $\rr{\TT C}[\ ]: (A^{*}\widetilde +\widetilde \top)\to \Bool\vdash^{\emptyset} \Bool$ can be written, up to $\eta$-equivalence, as $\rr{\TT C}[\ ]=\bb{ \Lambda Z.\lambda x_1x_{2}.}\rr{\TT C'}[\ ]$, with
% $\rr{\TT C'}[\ ]: (A^{*}\widetilde +\widetilde \top)\to \Bool\vdash^{\bb{x_{1}}:Z, \bb{x_{2}}:Z} \Bool$. As we can suppose $\rr{\TT C}[\ ]$ to be $\beta$-normal, by Lemma \ref{lemma1}, it must be $\rr{\TT C'}[\ ]\in \BB G_{1}$. Hence, by Lemma \ref{lemma2} we deduce that $\rr{\TT C}[\bb{u_{A}}]\simeq_{\beta\eta} \rr{\TT C}[\bb{v_{A}}]$. Moreover, as a closed term of type $\Bool$ must be either $\B t$ or $\B f$, from $\rr{\TT C}[\bb{u_{A}}]\simeq_{\beta\eta} \rr{\TT C}[\bb{v_{A}}]$ we can deduce
% $\rr{\TT C}[\bb{u_{A}}]\simeq_{\beta} \rr{\TT C}[\bb{v_{A}}]$.
%We have thus shown that $\bb{u_{A}}\simeq_{\Bool}^{\Nda}\bb{v_{A}}$.
%
\end{proof}

\begin{theorem}
The congruence $\simeq_{\Nat}^{\Nda}$ is undecidable.
%The $\mathrm{(CE)}$ problem for any of the systems $\Nda_{\alpha,\beta}$ is undecidable.
\end{theorem}
%\begin{proof}
%By Proposition \ref{prop:sbobis}, if $\simeq_{\Nat}^{\Nda}$ were decidable, we would be able to decide type inhabitation in $\Ndac$, which is impossible.\end{proof}

}

\end{document}